%% file: samplingMTL.tex
\documentclass[a4paper]{article}
\usepackage[pdftex]{graphicx}
\usepackage{rotating}
\usepackage{amsmath,amssymb,amsthm}
\usepackage{dsfont}
\usepackage{hyperref}
\usepackage{enumerate}

\usepackage{subfigure}

\usepackage{triotex}
\usepackage{clrscode}

\theoremstyle{plain}
\newtheorem{theorem}{Theorem}
\newtheorem{proposition}[theorem]{Proposition}
\newtheorem{lemma}[theorem]{Lemma}
\newtheorem{corollary}[theorem]{Corollary}

\theoremstyle{definition}
\newtheorem{definition}[theorem]{Definition}
\newtheorem{example}[theorem]{Example}

\newcommand{\nondig}{\Theta^{\text{snd}}}
\newcommand{\nonsi}{\Theta^{\text{dns}}}

\newcommand{\pp}{\mathsf{p}}
\newcommand{\qq}{\mathsf{q}}
\newcommand{\aaa}{\mathsf{a}}
\newcommand{\bb}{\mathsf{b}}
\newcommand{\xx}{\mathsf{x}}
\newcommand{\yy}{\mathsf{y}}
\newcommand{\zz}{\mathsf{z}}

\newcommand{\Scal}{\numset{S}}

\newcommand{\II}{\mathds{I}}

\newcommand{\alphabet}{\mathcal{P}}
\newcommand{\Ical}{\mathcal{I}}
\newcommand{\Dcal}{\mathcal{D}}

\newcommand{\bst}{\mathcal{B}\!\alphabet\!\timedomain}
\newcommand{\bsr}{\mathcal{B}\!\alphabet\!\reals}
\newcommand{\bsz}{\mathcal{B}\!\alphabet\!\integers}
\newcommand{\bsrd}{\mathcal{B}\!\alphabet\!\reals_{\delta}}

\newcommand{\bstgeneral}{\overline{\bst}}

\newcommand{\logictrue}{\top}
\newcommand{\logicfalse}{\bot}

\newcommand{\genop}[2]{\MTLoperator{\mathsf{O}}{#1}{}{#2}}

\newcommand{\PL}{\ensuremath{\mathrm{PL}}}
\newcommand{\MTL}{\ensuremath{\mathrm{MTL}}}
\newcommand{\LTL}{\ensuremath{\mathrm{LTL}}}
\newcommand{\flatMTL}{\ensuremath{\flat\MTL}}

\newcommand{\flatMTLplus}[1]{\ensuremath{{#1}\text{-}\flat\MTL}}

\newcommand{\expspace}{\languageclass{EXPSPACE}{}{}{}{}}

\DeclareRobustCommand{\SqBrackoperator}[4] 
{\ensuremath{%
	 \ifthenelse{\not \equal{#2}{}} {{#1}_{{#2}}^{#3}} {{#1}}
	 \ifthenelse{\not \equal{#4}{}} {\!\left[{#4}\right]} {}
  }%
}%

\newcommand{\modelstime}[1]{\models_{#1}}
\newcommand{\mdr}{\modelstime{\reals}}
\newcommand{\mdz}{\modelstime{\integers}}
\newcommand{\mdt}{\modelstime{\timedomain}}

\newcommand{\world}[2]{\ldoubsq#1\rdoubsq_{#2}}
\newcommand{\worldnb}[3]{\ldoubsq#1\rdoubsq_{#2}^{#3}}
\newcommand{\worldw}[2]{\langle\!\langle#1\rangle\!\rangle_{#2}}

\newcommand{\samp}[1]{\SqBrackoperator{\sigma}{\delta, z}{}{#1}}
\newcommand{\unsamp}[1]{\SqBrackoperator{\sigma}{\delta, z}{-1}{#1}}

\newcommand{\adapt}[1]{\SqBrackoperator{\eta}{\delta}{\reals}{#1}}
\newcommand{\unadapt}[1]{\SqBrackoperator{\eta}{\delta}{\integers}{#1}}
\newcommand{\unadaptinv}[1]{\SqBrackoperator{\eta}{\delta}{\integers^{-1}}{#1}}

\newcommand{\uap}{un\-der-ap\-prox\-i\-ma\-tion}
\newcommand{\oap}{over-ap\-prox\-i\-ma\-tion}

\newcommand{\becfMTL}[1]{\MTLoperator{\bigtriangleup}{}{}{#1}}
\newcommand{\becpMTL}[1]{\MTLoperator{\overleftarrow{\bigtriangleup}}{}{}{#1}}

\newcommand{\prevsamp}[1]{\mathrm{\Omega}(#1)}
\newcommand{\nextsamp}[1]{\mathrm{O}(#1)}
\newcommand{\prevsampdist}[1]{\mathrm{\omega}(#1)}
\newcommand{\nextsampdist}[1]{\mathrm{o}(#1)}

\DeclareMathOperator*{\lcm}{lcm}
\DeclareMathOperator*{\sinc}{sinc}

\newcommand{\Sys}{\Phi_{\mathsf{sys}}}
\newcommand{\system}{\phi_{\mathsf{sys}}}
\newcommand{\systembis}{\psi_{\mathsf{sys}}}
\newcommand{\prop}{\phi_{\mathsf{prop}}}
\newcommand{\verif}{\phi_{\mathsf{verif}}}

\newcommand{\overap}[1]{\SqBrackoperator{\mathrm{O}}{\delta}{}{#1}}

\newcommand{\underap}[1]{\SqBrackoperator{\mathrm{\Omega}}{\delta}{}{#1}}

\newcommand{\overmodel}[1]{#1^\mathrm{O}}
\newcommand{\undermodel}[1]{#1^\mathrm{\Omega}}

\newcommand{\tff}[2]{{#1} \stackrel{+}{\rightharpoondown} {#2}}
\newcommand{\tfp}[2]{{#1} \stackrel{-}{\rightharpoondown} {#2}}

\title{A Theory of Sampling \\ for Continuous-time Metric Temporal Logic\footnote{A preliminary version of this paper appeared in \cite{FR06,FPR08-FM08}.}}

\author{Carlo A. Furia and Matteo Rossi}

\date{\today}

\begin{document}

\maketitle

\begin{abstract}
This paper revisits the classical notion of sampling in the setting of real-time temporal logics for the modeling and analysis of systems.
The relationship between the satisfiability of Metric Temporal Logic (MTL) formulas over continuous-time models and over discrete-time models is studied.
It is shown to what extent discrete-time sequences obtained by sampling continuous-time signals capture the semantics of MTL formulas over the two time domains.
The main results apply to ``flat'' formulas that do not nest temporal operators and can be applied to the problem of reducing the verification problem for MTL over continuous-time models to the same problem over discrete-time, resulting in an automated partial practically-efficient discretization technique.
\end{abstract}

\newpage

\tableofcontents

\newpage

\input{introduction.tex}
\input{mtl.tex}
\input{sampling.tex}
\input{verification.tex}
\input{related.tex}
\input{conclusion.tex}

\paragraph{Acknowledgements.}
We thank the anonymous reviewers of the ACM Transactions on Computational Logic for their detailed comments.

\input{samplingMTL.bbl}


\end{document}

%% file: introduction.tex
\section{Introduction}
\label{sec:introduction}

Computer programs are inherently \emph{discrete} items, and they are typically modeled through techniques from the discrete mathematics domain.
If, however, one shifts from a computer-centric to a \emph{system-centric} view \cite{FMMR08}, physical elements, which are best described through \emph{continuous} signals, enter the picture and must be taken into account throughout the system development process.
This is the challenge that is at the core of the research on real-time and hybrid systems \cite{HS06}.
The challenge has two facets: {\em modeling} systems that integrate continuous and discrete components and {\em analyzing} properties of the integrated systems.

In this article we develop some techniques for the modeling and analysis of real-time systems with mixed continuous- and discrete-time components.
Our approach targets the well-known Metric Temporal Logic (MTL \cite{Koy90,AH93}) as formal notation, and it is is based on the classical notion of {\em sampling}.

Sampling is a widely-used technique in the engineering domain, in particular in signal processing and automatic control, whereby continuous-time signals are transformed in discrete-time counterparts that are more amenable to digital processing \cite{sampling-book}.
In systems where continuous- and discrete-time components interact, a \emph{sampler} constitutes the interface between these two classes of components, as it retains some \emph{partial} information of the continuous-time processes and passes it to the discrete-time parts (see Figure \ref{fig:sampler}).
\begin{figure}[htbp]
  \centering
  \includegraphics[scale=.75]{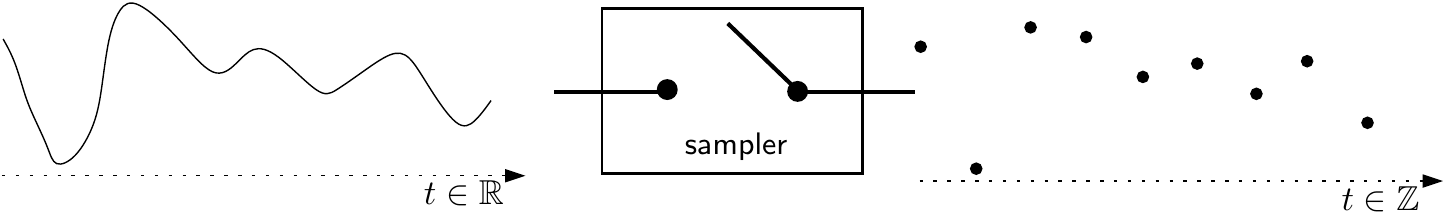}
  \caption{A system with a sampler.}
  \label{fig:sampler}
\end{figure}
The classical sampling theory determines qualitatively how much information is preserved in this discretization process, and when the continuous-time signal can be perfectly reconstructed solely from its discrete-time samplings.

The sampling approach described in this article borrows from these well-known ideas, but revisits them in the very different setting of formal modeling and analysis of systems with real-time temporal logics.\footnote{Section \ref{sec:sampling-theorem} discusses in some detail how the classical notion of sampling and the one presented here are related, though different.}

In our approach, the behavior of components is modeled by means of MTL formulas.
MTL formulas can be given a continuous-time or discrete-time semantics by interpreting them over sets of continuous- or discrete-time \emph{behaviors}\footnote{Also called (Boolean) \emph{signals} \cite{MNP06,HR04}.} respectively (see Section \ref{sec:MTL} for formal definitions).
Accordingly, MTL formulas can model either continuous- or discrete-time components.
The problem of providing a unified semantics is then solved by introducing simple \emph{syntactic transformations} to be applied to MTL formulas when moving their interpretation from continuous to discrete time, or \emph{vice versa}.
These transformations take into account the information that is preserved under sampling.
That is, given a continuous-time formula $\phi_\reals$, its (transformed) discrete-time counterpart $\phi_\integers$ is satisfied precisely by all discrete-time behaviors that are obtained by sampling the continuous-time behaviors satisfying $\phi_\reals$.
Information preservation requires however an additional requirement --- called non-Berkeleyness --- on the sampled continuous-time behaviors to ensure that they are sufficiently ``slow'' with respect to the speed of the sampling process.

In summary, the contribution of this article is twofold.
First, it introduces conditions that allow us to precisely relate the satisfiability of continuous-time MTL formulas to that of some suitable, ``sampled'', discrete-time counterparts.
Second, it exploits this relation to define an effective, albeit partial, automated analysis technique that can be used to prove (or disprove) properties of systems with continuous-time components by reduction to the (usually simpler) discrete-time case.
In this paper we do not deal with aspects regarding its implementation and performance in practice, which have been dealt with in related work \cite{FPR08-FM08,FPR08-ICFEM08,BFPR09-SEFM09}.
Rather, we focus on the mathematical concepts underlying the relation between continuous- and discrete-time semantics of MTL.

This article is structured as follows.
Section \ref{sec:MTL} introduces the MTL notation and its formal semantics, and discusses the expressiveness of some of its significant subsets.
Section \ref{sec:sampling} presents the notions of sampling and sampling invariance for MTL, and proves some fundamental results about significant subsets of the MTL language that are amenable to the sampling technique introduced beforehand, and hence are suitable to define a unified semantics.
Section \ref{sec:verification} shows how the results of Section \ref{sec:sampling} can be applied to the problem of automated verification of continuous-time systems described with MTL.
Finally, Section \ref{sec:related} provides an overview of related work, focusing on a few well-known approaches that are similar to ours; Section \ref{sec:conclusion} briefly concludes.

Let us remark that the mathematical distinction between continuous and merely dense time models does not impact the results of this paper.
Accordingly, we will essentially use the two terms as synonyms.

%% file: mtl.tex
\section{Metric Temporal Logic(s)} \label{sec:MTL}
The symbols $\integers$, $\rationals$, and $\reals$ denote the sets of integer, rational, and real numbers, respectively.
For a set $\Scal$, $\Scal_{\sim c}$ with $\sim$ one of $<, \leq, >, \geq$ and $c \in \Scal$ denotes the subset $\{ s \in \Scal \mid s \sim c \} \subseteq \Scal$; for instance $\integers_{\geq 0} = \naturals$ denotes the set of nonnegative integers (i.e., naturals).

An \emph{interval} $I$ of a set $\Scal$ is a convex subset $\langle l, u \rangle$ of $\Scal$ with $l, u \in \Scal$, $\langle$ one of $(, [$, and $\rangle$ one of $), ]$.
An interval is \emph{empty} iff it contains no points; an interval is \emph{punctual} (or singular) iff $l = u$ and the interval is closed (i.e., it contains exactly one point).
The \emph{length} of an interval is given by $|I| = \max(u-l, 0)$.
$-I$ denotes the interval $\langle -u, -l \rangle$, and $I \oplus t = t \oplus I$ denotes the interval $\langle t+l, t+u \rangle$, for any $t \in \Scal$.
For any numbers $x,y$ with $y > 0$, $x \pm \infty/y$ is defined to be $\pm \infty$.
We occasionally represent intervals by pseudo-arithmetic expressions such as $> x$, $\geq x$, $< x$, $\leq x$, and $= x$ for $(x, \infty)$, $[x, \infty)$, $[0, x)$, $[0, x]$ and $[x,x]$, respectively.
For simplicity, we sometimes relax the notation for unbounded intervals and represent them with square --- rather than round --- closing brackets.

\subsection{Behaviors} \label{sec:behaviors}
In this paper, $\timedomain$ denotes any of the two time domains $\reals$ and $\integers$.
It is not difficult to adapt most notions and results to their mono-infinite counterparts $\reals_{\geq 0}$ and $\naturals$, and possibly to other dense and discrete sets suitable to represent time domains \cite{Koy92}.
Also, let $\alphabet$ be a set of propositional letters.

\begin{definition}[Behaviors]
A \emph{(timed) behavior} over time domain $\timedomain$ and alphabet $\alphabet$ is a function $b: \timedomain \rightarrow 2^{\alphabet}$ which maps every time instant $t \in \timedomain$ to the set of propositions $b(t) \in 2^\alphabet$ that hold at $t$.
The set of all behaviors over time domain $\timedomain$ and alphabet $\alphabet$ is denoted by $\bstgeneral$.
\end{definition}

$b|_P$ is a behavior over alphabet $P \subseteq \alphabet$, denoting the projection of $b$ over $P$.
For a behavior $b$ over some \emph{dense} time domain $\timedomain$, let $\tau(b)$ denote the ordered (multi)set of its discontinuity points, that is $\tau(b) = \{ x \in \timedomain \mid b(x) \neq \lim_{t \rightarrow x^-}b(t) \text{, or } b(x) \neq \lim_{t \rightarrow x^+}b(t) \text{, or any of the two limits does not exist}\}$, \linebreak where each point that is both a right- and a left-discontinuity appears twice in $\tau(b)$.
When $\timedomain$ is a discrete set, $\tau(b)$ is defined to be the time domain $\timedomain$ itself.
If $\tau(b)$ is discrete, we can represent it as an ordered sequence (possibly unbounded to $\pm \infty$) of elements $\tau_i$ for $i \in \II$; it will be clear from the context whether we are treating $\tau(b)$ as a sequence or as a set.
Elements in $\tau(b)$ are called the \emph{change} (or \emph{transition}) instants of $b$.
$\tau(b)$ can be unbounded to $\pm \infty$ only if $\timedomain$ has the same property.

\paragraph{Non-Zenoness.}
Since one is typically interested only in behaviors that represent physically meaningful behaviors, it is common to assume some regularity requirements.
In particular, it is customary to assume \emph{non-Zenoness}, also called \emph{finite variability} \cite{HR04}.

\begin{definition}[Non-Zenoness]
A behavior $b \in \bstgeneral$ is non-Zeno iff $\tau(b)$ has no accumulation points.
The set of all non-Zeno behaviors is denoted by $\bst$.
\end{definition}

Notice that discrete-time behaviors are trivially non-Zeno.
Also, it should be clear that every non-Zeno behavior can be represented through a canonical countable sequence of adjacent intervals of $\timedomain$ such that $b$ is constant on every such interval.
Namely, for $b \in \bst$, $\iota(b)$ is an ordered sequence of intervals $\iota(b) = \{ I_i = \langle^i l_i, u_i \rangle^i \}$ for $i \in \II$ such that:
\begin{enumerate}
\item \emph{(cardinality of $\iota(b)$)} $\II$ is an interval of $\integers$ with cardinality $|\tau(b)| + 1$ (in particular, $\II$ is finite iff $\tau(b)$ is finite, otherwise $\II$ is denumerable);
\item \emph{(partitioning of $\timedomain$)} the intervals in $\iota(b)$ form a partition of $\timedomain$;
\item \emph{(intervals change at transition points)} for all $i \in \II$ we have $\tau_i = u_i = l_{i+1}$;
\item \emph{($b$ constant over intervals)} for all $i \in \II$, for all $t_1, t_2 \in I_i$ we have $b(t_1) = b(t_2)$.
\end{enumerate}
Note that $\iota(b)$ is unique for any fixed set $\tau(b)$ or, in other words, is unique up to translations of interval indices.
Transitions at instants $\tau_i$ corresponding to singular intervals $I_i$ are called \emph{pointwise} (or \emph{punctual}) transitions.

\paragraph{Non-Berkeleyness.}
Some of the results of this paper will require a stronger regularity requirement than non-Zenoness, named ``non-Berkeleyness'' \cite{FPR08-FM08}.

\begin{definition}[Non-Berkeleyness] \label{def:non-berkeley}
A \emph{behavior} $b \in \bst$ is \emph{non-Berkeley} for $\delta \in \reals_{> 0}$ iff every maximal constancy interval contains a closed interval of size $\delta$.
The set of all behaviors in $\bst$ that are non-Berkeley for $\delta$ is denoted by $\bst_\delta$; with the notation introduced above, it is $\bst_\delta = \{b \in \bst \mid \forall I \in \iota(b): \exists t \in I: [t,t+\delta] \subseteq I\}$.
A behavior that is \emph{not} non-Berkeley for any positive $\delta$ is called \emph{Berkeley}.
\end{definition}

Any behavior where some proposition holds at an isolated point $t$ is Berkeley: any $\delta > 0$ is such that $[t, t+\delta] \not\subseteq [t,t]$.

\subsection{MTL: Syntax and Semantics}
This section defines formally the syntax and semantics of MTL.

\subsubsection{MTL Syntax} \label{sec:mtl-syntax}
In this paper only \emph{propositional} temporal logics are considered; correspondingly, the elementary building block of temporal logic formulas is defined.
\begin{definition}[Propositional formulas]
Propositional formulas $\pi \in \PL$ are defined by the grammar $\pi ::= \pp \mid \neg \pp \mid \pi_1 \wedge \pi_2 \mid \pi_1 \vee \pi_2$ --- for $\pp \in \alphabet$ --- as Boolean combinations of propositional letters.
\end{definition}

MTL formulas are obtained by combining propositional formulas with the \emph{bounded until} $\untilMTL{I}{}$ metric modality, as well as its past counterpart \emph{bounded since} $\sinceMTL{I}{}$.
We assume a \emph{negation normal form} (NNF) syntax, where negations are pushed down to atomic propositions, as this will simplify the presentation of the results.
Correspondingly, \emph{bounded release} $\relMTL{I}{}$ and \emph{bounded trigger} $\redMTL{I}{}$ operators --- duals to the \emph{until} and \emph{since} operators, respectively --- are introduced as primitive modalities.
\begin{definition}[MTL formulas]
MTL formulas for a time domain $\timedomain$ are defined by the grammar:
\begin{equation*}
\phi ::=  \pi \mid \phi_1 \wedge \phi_2 \mid \phi_1 \vee \phi_2 \mid \untilMTL{I}{\phi_1, \phi_2} \mid \sinceMTL{I}{\phi_1, \phi_2} \mid \relMTL{I}{\phi_1, \phi_2} \mid \redMTL{I}{\phi_1, \phi_2}
\end{equation*}
where $\pi \in \PL$ ranges over propositional formulas and $I$ ranges over (possibly unbounded) intervals of the time domain $\timedomain$ with endpoints in $\timedomain \cap \rationals \cup \{\pm \infty\}$ (notice that negative endpoints are allowed).
\end{definition}
Henceforth, we will drop interval $I$ in modalities when it is $[0, +\infty)$.

The results of this paper are focused on the \emph{flat} subset \flatMTL{} of MTL, whose formulas do not nest temporal operators.\footnote{Different notions of flatness for (metric) temporal logic have been introduced in the literature \cite{Dam99,CC00,BMOW07}.}
\begin{definition}[Flat MTL formulas]
\flatMTL{} formulas for a time domain $\timedomain$ are defined by the grammar:
\begin{equation*}
\phi ::=  \pi \mid \phi_1 \wedge \phi_2 \mid \phi_1 \vee \phi_2 \mid \untilMTL{I}{\pi_1, \pi_2} \mid \sinceMTL{I}{\pi_1, \pi_2} \mid \relMTL{I}{\pi_1, \pi_2} \mid \redMTL{I}{\pi_1, \pi_2}
\end{equation*}
where $\pi,\pi_1,\pi_2 \in \PL$ range over propositional formulas and $I$ ranges over (possibly unbounded) intervals of the time domain $\timedomain$ with endpoints in $\timedomain \cap \rationals \cup \{\pm \infty\}$.
\end{definition}

In the remainder of the paper, the following other MTL subsets will be needed.
\begin{itemize}
\item \emph{\LTL{}} is the MTL subset where all intervals $I$ are $[0,+\infty)$ (i.e., all operators are \emph{qualitative}).
\item \emph{\flatMTLplus{\Upsilon}}, with $\Upsilon$ any given set of MTL formulas, is the MTL subset defined by the same grammar as \flatMTL{}, except that $\pi$ is allowed to range over $\PL \cup \Upsilon$.
\item An MTL formula is \emph{discrete-endpoint} if all its intervals have endpoints in $\integers \cup \{\pm \infty\}$.
\item An MTL formula is \emph{dense-endpoint} if all its intervals have endpoints in $\reals \cup \{\pm \infty\}$.
It is clear that any MTL formula is dense-endpoint; we will use this redundant terminology whenever useful to characterize formulas to be interpreted over a dense time domain, as opposed to a discrete one.
\end{itemize}

\subsubsection{MTL Semantics} \label{sec:mtl-semantics}
We define MTL semantics parametrically with respect to the time domain $\timedomain$.
\begin{definition}[MTL semantics] \label{def:mtl-semantics}
Let $b \in \bst$ be a behavior over $\alphabet$ and time domain $\timedomain$.
For $t \in \timedomain$, MTL semantics is defined recursively as follows.\footnote{In this paper, the notation $b(t) \mdt \phi$ replaces the more common $b,t \mdt \phi$.} \\
\begin{tabular}{l c l}
  $b(t) \mdt \pp$  & \ \ \ iff\ \ \ & $\pp \in b(t)$ \\
  $b(t) \mdt \neg \pp$  & \ \ \ iff\ \ \ & $\pp \not\in b(t)$ \\
  $b(t) \mdt \phi_1 \wedge \phi_2$  & \ \ \ iff\ \ \ & $b(t) \mdt \phi_1$ and $b(t) \mdt \phi_2$ \\
  $b(t) \mdt \phi_1 \vee \phi_2$  & \ \ \ iff\ \ \ & $b(t) \mdt \phi_1$ or $b(t) \mdt \phi_2$ \\
  $b(t) \mdt \untilMTL{I}{\phi_1, \phi_2}$  & \ \ \ iff\ \ \ & $\exists d \in I$ s.t.: $t+d \in \timedomain$, $b(t+d)\mdt\phi_2$, and \\
                                   & & $\forall t' \in [0,d) \oplus t \cap \timedomain$ it is $b(t') \mdt\phi_1$ \\
  $b(t) \mdt \sinceMTL{I}{\phi_1, \phi_2}$  & \ \ \ iff\ \ \ & $\exists d \in I$ s.t.: $t-d \in \timedomain$, $b(t-d)\mdt\phi_2$, and \\
                                   & & $\forall t' \in -[0,d) \oplus t \cap \timedomain$ it is $b(t') \mdt\phi_1$ \\
  $b(t) \mdt \relMTL{I}{\phi_1, \phi_2}$  & \ \ \ iff\ \ \ & $\forall d \in I$ s.t. $t+d \in \timedomain$ it is: $b(t+d)\mdt\phi_2$ or \\
                                   & & $\exists t' \in [0,d) \oplus t \cap \timedomain$ s.t.~$b(t') \mdt\phi_1$ \\
  $b(t) \mdt \redMTL{I}{\phi_1, \phi_2}$  & \ \ \ iff\ \ \ & $\forall d \in I$ s.t. $t-d \in \timedomain$ it is: $b(t-d)\mdt\phi_2$ or \\
                                   & & $\exists t' \in -[0,d) \oplus t \cap \timedomain$ s.t.~$b(t') \mdt\phi_1$ \\
\end{tabular} \\
If $b(t) \mdt \phi$ holds for all $t \in \timedomain$ we write $b \mdt \phi$.
\end{definition}

We denote by $\world{\phi}{\timedomain}$ (respectively $\worldnb{\phi}{\timedomain}{\delta}$) the set of all non-Zeno (respectively non-Berkeley for $\delta$) models of formula $\phi$ over $\timedomain$, i.e., $\world{\phi}{\timedomain} \triangleq \{b \in \bst \mid b \mdt \phi \}$ (respectively $\worldnb{\phi}{\timedomain}{\delta} \triangleq \{b \in \bst_\delta \mid b \mdt \phi \}$).
If $\world{\phi}{\timedomain}$ is empty, $\phi$ is called $\timedomain$-unsatisfiable, and $\timedomain$-satisfiable otherwise.
If $\world{\phi}{\timedomain}$ coincides with $\bst$, $\phi$ is called $\timedomain$-valid.
Similar definitions are assumed for $\timedomain^\delta$-satisfiability and $\timedomain^\delta$-validity, with respect to $\worldnb{\phi}{\timedomain}{\delta}$.
For $b \in \bst$, we define the derived behavior $b_\phi$ that represents the truth value of $\phi$ over $b$ as:
\begin{equation*}
  b_{\phi}(t) = 
  \begin{cases}
	 b(t) \cup \{ \phi \}  &  \text{if } b(t) \mdt \phi \\
    b(t)  &  \text{otherwise.}
  \end{cases}
\end{equation*}
For propositional letters $\aaa, \bb$, $b^{\aaa\setminus\bb}$ denotes the behavior obtained from $b$ by renaming $\aaa$ into $\bb$.

Notice that \MTL{} is closed under complement, even if this is not apparent in the definition of its syntax.
More precisely, one can check that $b(t) \not\mdt \untilMTL{I}{\phi_1, \phi_2}$ holds if and only if $b(t) \mdt \relMTL{I}{\neg \phi_1, \neg \phi_2}$ does, thus providing an indirect definition of negation.
A similar relation holds for \emph{since} with respect to \emph{trigger}.

Definition \ref{def:mtl-semantics} considers the basic modalities in their \emph{non-strict} versions as, for instance, $\untilMTL{}{\phi_1, \phi_2}$ requires $\phi_1$ to hold at the current instant; i.e., it constrains the present as well as the strict future.
Also, a \emph{global satisfiability semantics} is assumed, where $b \mdt \phi$ entails that $\phi$ holds at all time instants $t \in \timedomain$.
This is different than the more common \emph{initial satisfiability semantics} $\mdt^{\mathrm{init}}$ where $b \mdt^{\mathrm{init}} \phi$ is defined as simply $b(0) \mdt \phi$.
Section \ref{sec:relations} discusses the impact of these choices on expressiveness.

\subsubsection{Derived Operators and Variants}
Standard abbreviations are assumed, such as for $\logictrue$, $\logicfalse$, $\Rightarrow$, and $\Leftrightarrow$.

It is also customary to introduce a number of derived temporal operators; those used in this paper are listed in Table \ref{tab:mtl-derived}.
Let us remark that the definitions of Table \ref{tab:mtl-derived} do not nest temporal operators, hence they define $\flatMTL$ formulas if their arguments are propositional formulas.

The first set of derived operators are the quantitative versions of the well-known \emph{eventually} $\diamondMTL{}{}$ and \emph{globally} $\boxMTL{}{}$ modalities of classic (qualitative) linear temporal logic.
On the other hand, $\Alw{\phi}$ declares $\phi$ to hold \emph{always}, i.e., at all time instants in the future and in the past, whereas $\Som{\phi}$ declares $\phi$ to hold \emph{sometimes}.

The second set of derived operators are the \emph{nowon} $\nowonMTL{}$ modality and its variant $\becfMTL{}$, with their past counterparts \emph{uptonow} $\uptonowMTL{}$ and $\becpMTL{}$.
Over dense-time non-Zeno behaviors, $\nowonMTL{\phi}$ holds at $t$ whenever there is a non-empty open interval $E = (0, \epsilon)$ (with $\epsilon > 0$) such that $\phi$ holds continuously over $t \oplus E$.
On the other hand, $\becfMTL{\phi}$ holds at $t$ whenever $\phi$ holds \emph{nowon} or $\phi$ holds precisely at $t$.
These operators are useful only over dense time, as they can be seen to be trivially equivalent to their arguments over discrete time.

Finally, the last set of derived operators introduce so-called \emph{matching variants} \cite{FR07-FORMATS07} of the basic \emph{until} and \emph{release} modalities.
For instance \emph{matching until} $\untilMMTL{ }{\phi_1, \phi_2}$ requires both arguments $\phi_1$ and $\phi_2$ to hold together at some future instant, whereas $\untilMTL{ }{\phi_1, \phi_2}$ demands only $\phi_2$ to hold at some future instant.
The next section discusses the impact of these variants on expressiveness.

\begin{table}[htb]
\begin{center}
  \begin{tabular}{|c @{$\quad \triangleq \quad$} c|}
	 \hline
    \textsc{Operator}        & \textsc{Definition} \\
    \hline
	 $\diamondMTL{I}{\phi}$    &   $\untilMTL{I}{\logictrue, \phi}$ \\
	 $\diamondPMTL{I}{\phi}$    &   $\sinceMTL{I}{\logictrue, \phi}$ \\
	 $\boxMTL{I}{\phi}$    &   $\relMTL{I}{\logicfalse, \phi}$ \\
	 $\boxPMTL{I}{\phi}$    &   $\redMTL{I}{\logicfalse, \phi}$ \\
    $\Alw{\phi}$  &  $\boxPMTL{}{\phi} \wedge \boxMTL{}{\phi}$ \\
    $\Som{\phi}$  &  $\diamondPMTL{}{\phi} \vee \diamondMTL{}{\phi}$ \\
    \hline
    $\nowonMTL{\phi}$     &    $\untilMTL{> 0}{\phi, \logictrue} \vee (\neg \phi \wedge \relMTL{> 0}{\phi, \logicfalse})$ \\
    $\uptonowMTL{\phi}$     &    $\sinceMTL{> 0}{\phi, \logictrue} \vee (\neg \phi \wedge \redMTL{> 0}{\phi, \logicfalse})$ \\
    $\becfMTL{\phi}$     &    $\phi \vee \nowonMTL{\phi}$ \\
    $\becpMTL{\phi}$     &    $\phi \vee \uptonowMTL{\phi}$ \\
   \hline
    $\untilMMTL{I}{\phi_1, \phi_2}$  &  $\untilMTL{I}{\phi_1, \phi_2 \wedge \phi_1}$ \\
    $\sinceMMTL{I}{\phi_1, \phi_2}$  &  $\sinceMTL{I}{\phi_1, \phi_2 \wedge \phi_1}$ \\
    $\relMMTL{I}{\phi_1, \phi_2}$  &  $\relMTL{I}{\phi_1, \phi_2 \vee \phi_1}$ \\
    $\redMMTL{I}{\phi_1, \phi_2}$  &  $\redMTL{I}{\phi_1, \phi_2 \vee \phi_1}$ \\
    \hline
  \end{tabular}
\caption{MTL derived temporal operators}
\label{tab:mtl-derived}
\end{center}
\end{table}

The value of propositional formulas change at most every $\delta$ time units over non-Berkeley behaviors $\bsrd$; more precisely the following holds.
\begin{lemma} \label{lemma:changepoints}
Let $b \in \bsrd$, $t \in \reals$, and $\pi \in \PL$, such that $b(t) \mdr \pi$.
There exist $c_n, c_p \in \reals$ with $c_n - c_p \geq \delta$ and $c_p \leq t \leq c_n$ such that: (1) $b(t') \mdr \pi$ for all $t' \in (c_p, c_n)$; (2) $b(c_n) \mdr \boxMTL{(0,\delta)}{\neg \pi} \vee \boxMTL{}{\pi}$; and (3) $b(c_p) \mdr \boxPMTL{(0,\delta)}{\neg \pi} \vee \boxPMTL{}{\pi}$.
If in particular $c_n - c_p = \delta$ then also $b(c_n) \mdr \pi$ and $b(c_p) \mdr \pi$.
\end{lemma}
\begin{proof}
The proof follows easily from Definition \ref{def:non-berkeley}, which entails that non-Berkeley behaviors $b \in \bsrd$ are piecewise-constant functions of time whose discontinuities are at least $\delta$ time units apart.
\end{proof}

\subsection{Relations with Other Metric Temporal Logics} \label{sec:relations}
This section discusses expressiveness, decidability, and complexity results about MTL as has been introduced above.

\subsubsection{Expressiveness}
When defining the semantics of MTL formulas, several different choices are possible.

\paragraph{Global vs.~initial satisfiability.}
First of all, notice that initial satisfiability is unambiguous only for mono-infinite time domains \cite{PP04}.
For such domains, it is clear that the global satisfiability semantics can be reduced to local satisfiability, as $b \mdt \phi$ holds if and only if $b(0) \mdt \Alw{\phi}$ does.
Conversely, local satisfiability is also reducible to global satisfiability, as for instance $b \models_{\reals_{\geq 0}} \boxPMTL{> 0}{\logicfalse} \Rightarrow \phi$ is equivalent to $b(0) \models_{\reals_{\geq 0}} \phi$ for the mono-infinite time domain $\reals_{\geq 0}$, where $\boxPMTL{> 0}{\logicfalse}$ holds only at time $0$.
Therefore the two definitions of satisfiability are essentially equivalent for generic MTL formulas.

However, global satisfiability is significantly more expressive than initial satisfiability for \emph{flat} $\flatMTL$ formulas \cite{FR07-FORMATS07}.
In particular, the expressiveness of the flat fragment is non-trivial under such global semantics as it corresponds to an implicit nesting of a qualitative temporal operator over the simpler initial satisfiability semantics.
This entails that most common (real-time) properties --- such as (bounded) response and (bounded) invariance \cite{Koy90} --- can be easily expressed with flat formulas under the global satisfiability semantics.
This is the main reason for adopting such a semantics in this paper whose results are focused on the flat fragment of MTL.

\paragraph{Flat vs.~nesting.}
The syntactic restriction of flatness is also a semantic restriction, i.e., $\flatMTL$ is strictly less expressive than full MTL.
This is the case not only for dense time (which has been proved in \cite{FR07-FORMATS07}) but also for discrete time (which has been proved in \cite{EW96,TW04,KS05,DS02} already for qualitative temporal logic), and regardless of whether a global or initial satisfiability semantics is assumed.

On the other hand, if we consider the weaker requirement of inter-reducibility of the satisfiability problems over global satisfiability, $\flatMTL$ is as powerful as full MTL.
In other words, given any MTL formula $\phi$, it is possible to build a flat formula $\phi' \in \flatMTL$ which is globally satisfiable if and only if $\phi$ is.
In general, $\phi'$ ``flattens'' $\phi$ by introducing additional propositional letters that are equivalent to matching nested sub-formulas in $\phi$, as shown in the following.
\begin{example}
Let $\phi = \pp \Rightarrow \diamondMTL{< 3}{\nowonMTL{\boxMTL{= 2}{\qq}}}$.
Let us introduce the auxiliary propositions $\aaa_1$ and $\aaa_2$ defined as equivalent to $\boxMTL{=2}{\qq}$ and $\nowonMTL{\aaa_1}$ respectively.
Hence, the derived \emph{flat} formula
\begin{equation*}
\phi' \quad = \quad \left(\pp \Rightarrow \diamondMTL{< 3}{\aaa_2}\right) \ \wedge\  \left(\aaa_1 \Leftrightarrow \boxMTL{=2}{\qq}\right) \ \wedge\  \left(\aaa_2 \Leftrightarrow \nowonMTL{\aaa_1}\right)
\end{equation*}
is equi-satisfiable to $\phi$ under the global satisfiability semantics.
\end{example}

Details of this straightforward idea are shown in \cite{Fur07,MP07} for dense time models, but it should be clear that a similar result can be proved for discrete time as well.

Let us finally consider dense-time behaviors that are non-Berkeley.
In this case, the expressiveness gap between flat and nesting formulas still exists \cite{FR07-FORMATS07}.
On the other hand, ``flattening'' is more intricate and cannot be done as with generic behaviors without breaking non-Berkeleyness as shown in the following example and discussed at greater length in Section \ref{sec:generalizations}.
\begin{example} \label{ex:nonflattable}
MTL formula $\psi = \Som{}\left(\uptonowMTL{\neg \pp} \wedge \nowonMTL{\pp}\right)$ describes behaviors where there exists a transition of proposition $\pp$ from false to true.
$\psi$ is satisfiable over non-Berkeley behaviors $\bsr_\delta$ for any positive $\delta$.
However, consider the flattening $\overline{\psi}$ of $\psi$ built according to the procedure described above.
\begin{equation*}
  \overline{\psi} \quad = \quad \Som{\aaa} \ \wedge \ \left(\aaa \Leftrightarrow \uptonowMTL{\neg \pp} \wedge \nowonMTL{\pp} \right)
\end{equation*}
Any behavior $b \in \bsr$ such that $b \mdr \overline{\psi}$ requires $\aaa$ to hold at some instant $t$.
However, sub-formula $\aaa \Leftrightarrow \uptonowMTL{\neg \pp} \wedge \nowonMTL{\pp}$ forces $\aaa$ to hold only exactly at the transition points of $\pp$, pointwisely: any such $b$ is Berkeley because $\uptonowMTL{\neg \pp} \wedge \nowonMTL{\pp}$ holds only at isolated points.
Hence, $\psi$ and $\overline{\psi}$ are not equi-satisfiable over non-Berkeley behaviors for any $\delta > 0$.
\end{example}

\paragraph{Strictness and matchingness.}
The semantics of an \emph{until} formula with arguments $\phi_1, \phi_2$ requires the first argument $\phi_1$ to hold over an interval $J = \langle 0, d \rangle \oplus t$ from current instant $t$.
$J$ can be taken to be open, half-open (with the left or right end-point included), or closed.
Correspondingly four variants of \emph{until} are possible.
Each of them is labeled \emph{strict} if $J$ is open to the left and \emph{non-strict} otherwise; and \emph{matching} if $J$ is closed to the right and \emph{non-matching} otherwise \cite{FR07-FORMATS07}.
The most common variant of the \emph{until} operator is strict and non-matching, as it is simple to see that the three other variants are reducible to it.
On the contrary, this paper adopts a non-strict non-matching \emph{until} as basic operator, as the presentation of the results is more natural with non-strict operators.

In related work \cite{FR07-FORMATS07}, we proved that all variants carry the same expressive power for MTL over dense- and discrete-time behaviors.
On the contrary, strict \emph{until} is more expressive than non-strict \emph{until} for \emph{flat} \flatMTL{} formulas.\footnote{\cite{FR07-FORMATS07} proves this for dense-time behaviors, but the same can be seen to hold over discrete-time behaviors as well.}

\subsubsection{Decidability and Complexity}
It is well-known that full MTL is undecidable over (non-Zeno) dense-time behaviors \cite{AH93}.
The same holds for flat \flatMTL{} as its satisfiability problem is inter-reducible to the same problem for full MTL.

On the contrary, MTL becomes fully decidable over discrete time, with \expspace{}-complete complexity \cite{AH93}.
MTL is also fully decidable over non-Berkeley dense-time behaviors $\bst_\delta$ for any fixed $\delta$, with the same complexity as over discrete time \cite{FR08-FORMATS08}.


%% file: sampling.tex
\section{Sampling Invariance} \label{sec:sampling}
Throughout this section we assume $\reals$ as dense (and continuous) time domain, and $\integers$ as discrete time domain.
It should be noted, however, that nearly all definitions and results can be adapted with little effort to work with different pairs of dense and discrete time domains as well, most notably the nonnegative reals and the naturals.

\subsection{Definitions}

\subsubsection{Sampling Functions}
A \emph{sampling function} is a mapping between dense-time behaviors and discrete-time behaviors such that the latter are obtained by ``sampling'' --- in some sense --- the values of the former.
We use $\varsigma_{z,\delta}$ to denote a generic sampling function that is parametric with respect to a sampling period $\delta$ and an origin $z$.

The \emph{canonical sampling} is a particular sampling function that models an idealized sampling process where a discrete-time behavior is obtained from a dense-time behavior by observing it at all instants corresponding to integer multiples of a chosen period $\delta$. 

\begin{definition}[Canonical sampling of a behavior]
Let $b \in \bsr$ be a dense-time behavior, $\delta \in \reals_{> 0}$ a positive real, and $z \in \reals$ a basic offset.
The \emph{canonical sampling} $\samp{b}$ of $b$ is the discrete-time behavior in $\bsz$ defined by:
\begin{equation*}
  \forall k \in \integers:\quad \samp{b}(k) \ =\  b(z + k\delta)
\end{equation*}
We call $\delta$ the \emph{sampling period} and $z$ the \emph{origin} of the sampling.
Note that $\samp{}$ is onto and total,\footnote{That is, it is defined for every $b \in \bsr$.} for any $\delta, z$.

Conversely, given a discrete-time behavior $d \in \bsz$, $\unsamp{z}$ is the set of all dense-time behaviors such that their sampling is $d$.
\begin{equation*}
  \unsamp{d} \ = \ \{ b \in \bsr \mid d = \samp{b} \}
\end{equation*}
\end{definition}

\subsubsection{On Dense- vs.~Discrete-Time Semantics}
Consider some MTL formula $\phi$ that can be interpreted over both dense- and discrete-time behaviors.
Its semantics is characterized by its dense-time models $\world{\phi}{\reals}$ on the one hand, and by its discrete-time models $\world{\phi}{\integers}$ on the other hand.
These two sets correspond to two different semantics for the \emph{same} formula.
The fact that the discrete time domain is a subset of the dense time domain prompts us to investigate the existence of a general relation linking the two sets $\world{\phi}{\reals}$ and $\world{\phi}{\integers}$.
More precisely, we seek simple conditions under which elements in $\world{\phi}{\integers}$ are precisely those obtained from elements in $\world{\phi}{\reals}$ by applying the sampling function $\sigma_{\delta, z}$ for some $\delta, z$.

This ideal requirement must be relaxed to some extent to be achievable in practice, for a number of reasons that are outlined informally in the following example.
\begin{example} \label{ex:densevsdiscrete}
There are three fundamental discrepancies between discrete- and dense-time semantics that must be accommodated to reconcile them according to the notion of sampling.

The first has to do with differences in terms of \emph{time units}.
Consider for instance formula $\boxMTL{\leq 2}{\pp}$; when interpreted over dense time, it states that $\pp$ holds for 2 time units.
If we switch to a discrete-time interpretation and consider a sampling period of, say, $\delta = 3/10$, we would like the formula to refer to the same ``sampled'' interval.
Hence, it should be changed to $\boxMTL{\leq 20/3}{\pp}$ because the dense-time interval of length 2 becomes a discrete-time interval containing $2 / (3/10)$ sampling instants.

However, $\boxMTL{\leq 20/3}{\pp}$ cannot yet be interpreted over discrete time, as $20/3$ is not an integer; this shows a discrepancy in terms of \emph{granularity} between dense and discrete sets.
Of course, this problem can be solved by rounding the rational value to the nearest integer value, by taking its floor $6$ or its ceiling $7$.
More precisely, whether to round up or down is decided in order to have a \emph{conservative} approximation of the semantics.
Intuitively, this means that intervals in ``universal'' formulas such as $\boxMTL{I}{\pp}$ are rounded down --- thus shrinking the interval into a smaller one ---, whereas intervals in ``existential'' formulas such as $\diamondMTL{I}{\pp}$ are rounded up --- thus expanding the interval into a larger one.

A similar granularity problem arises when interpreting a discrete-endpoint formula over dense time.
Consider the example of formula $\diamondMTL{[1,2]}{\pp}$ that requires $\pp$ to hold one or two (discrete) time units in the future.
In terms of dense-time units, $\pp$ must occur over the interval $[\delta, 2\delta] = [3/10, 3/5]$.
However, the formula must hold also \emph{in between} sampling instants when interpreted over dense-time behaviors.
We will show that this feature of the dense-time semantics can be accommodated by expanding symmetrically the scaled interval into $[(1-1)\delta, (2+1)\delta] = [0, 9/10]$.

The last subtlety has to do with the \emph{change speed} of dense-time behaviors with respect to the sampling period.
Consider behavior $b$ over proposition $\pp$ such that $\pp$ holds for less than $\delta$ time units, say over $[\delta/4,\delta/2]$.
Formula $\Som{\pp}$ is clearly satisfied by $b$ over discrete time.
However, any sampling $\samp{b}$, for any $z \in (-\delta/2, \delta/4) \cup (\delta/2, \infty)$, does not have any sampling instant within $[\delta/4, \delta/2]$, and formula $\Som{\pp}$ is not satisfied by any such $\samp{b}$.
This shows that only dense-time behaviors where state changes are sufficiently sparse can guarantee that formula satisfaction is preserved while moving to a sampled discrete-time semantics.
\end{example}

The discrepancies outlined above are bridged by introducing suitable notions.
The concept of slowly-changing behavior is captured by the non-Berkeleyness constraint, introduced in Section \ref{sec:behaviors}.
The following notion of \emph{adaptation function} formalizes instead changes to intervals in MTL formulas, which take discrepancies between time units and granularities into account.

\begin{definition}[Adaptation] \label{def:adaptation}
A \emph{$\reals$-to-$\integers$ adaptation} is a mapping from dense-endpoint to discrete-endpoint MTL formulas; a \emph{$\integers$-to-$\reals$ adaptation} is a mapping from discrete-endpoint to dense-endpoint MTL formulas.
\end{definition}

\subsubsection{Sampling Invariance} \label{sec:sampling-invariance}
We can finally introduce the definition of \emph{sampling invariance} over non-Berkeley behaviors, which captures appropriately a notion of equivalence under sampling of models of MTL formulas.

\begin{definition}[Sampling invariance] \label{def:samplinginvariance}
Let $\phi$ be an MTL formula over alphabet $\alphabet$; $\upsilon^\reals, \upsilon^\integers$ a $\reals$-to-$\integers$ and $\integers$-to-$\reals$ adaptation, respectively; $\delta$ a sampling period; and $\varsigma_{\delta, z}$ a sampling function.
  \begin{itemize}
	 \item $\phi$ is \emph{closed under sampling} (c.u.s.) iff for any non-Berkeley behavior $b \in \bsrd$ and any origin $z$:
		\begin{equation*}
		  b \in \worldnb{\phi}{\reals}{\delta} \qquad \text{ implies } \quad
                 \varsigma_{\delta,z}[b] \in \world{\upsilon^{\reals}[\phi]}{\integers}
		\end{equation*}

    \item $\phi$ is \emph{closed under inverse sampling} (c.u.i.s.) iff for any discrete-time behavior $b \in \bsz$ and any origin $z$:
		\begin{equation*}
		  b \in \world{\phi}{\integers} \qquad \text{ implies } \quad
        \forall b' \in \varsigma_{\delta, z}^{-1}[b] \cap \bsrd{} \text{ it is } b' \in \worldnb{\upsilon^{\integers}[\phi]}{\reals}{\delta}
		\end{equation*}

		\item $\phi$ is \emph{sampling invariant} (s.i.) iff it is c.u.s.\ if it is a dense-endpoint formula \emph{and} it is c.u.i.s.\ if it is a discrete-endpoint formula.
  \end{itemize}
\end{definition}

Definition \ref{def:samplinginvariance} depends on several parameters: $\reals$-to-$\integers$ and $\integers$-to$\reals$ adaptations, a sampling period $\delta$, and a sampling function $\varsigma_{\delta, z}$.
In the following we will use the expression ``sampling invariance (c.u.s.\ or c.u.i.s) with respect to'' to highlight a particular choice for the parameters (when they are not obvious from the context).

\subsection{Illustrative Examples}
\begin{figure}[!tp]
\centering
\subfigure[Behaviors $c_1, d_1$]{
\includegraphics[width=.9\textwidth]{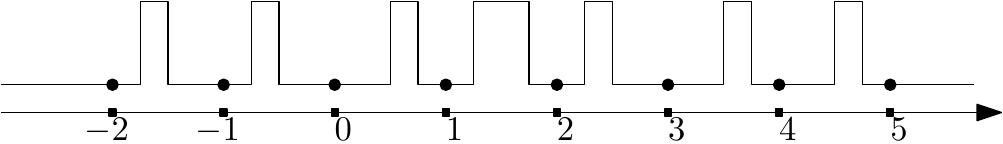}
\label{fig:ex1}
}
 \subfigure[Behaviors $c_2, d_2$]{
\includegraphics[width=.9\textwidth]{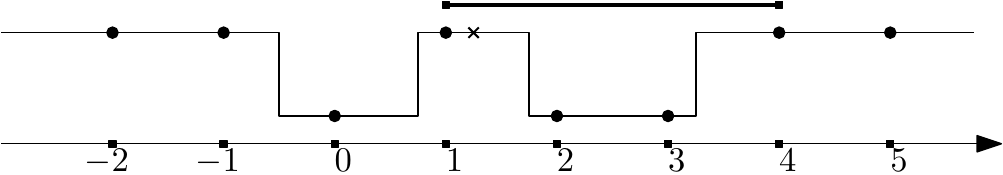}
\label{fig:ex2}
}
\subfigure[Behaviors $c_3, d_3$]{
\includegraphics[width=.9\textwidth]{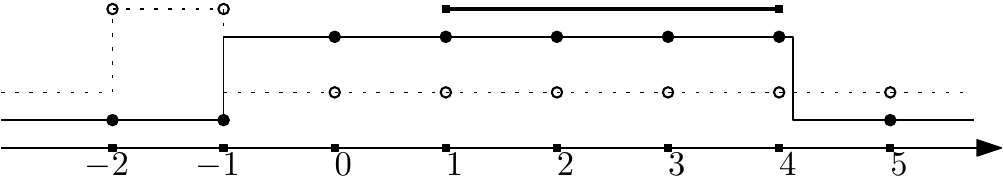}
\label{fig:ex3}
}
\subfigure[Behaviors $c_4, d_4$]{
\includegraphics[width=.9\textwidth]{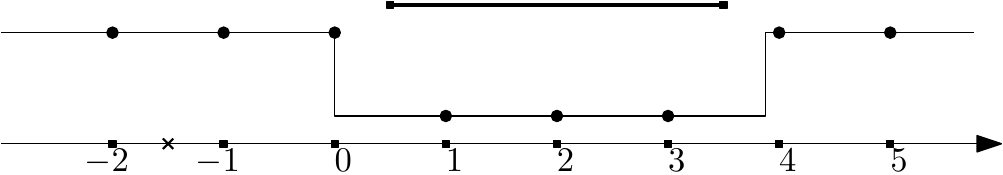}
\label{fig:ex4}
}
\subfigure[Behaviors $c_5, d_5$]{
\includegraphics[width=.9\textwidth]{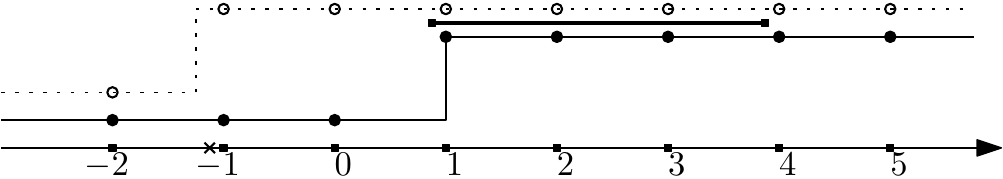}
\label{fig:ex5}
}
\subfigure[Behaviors $c_6, d_6$]{
\includegraphics[width=.9\textwidth]{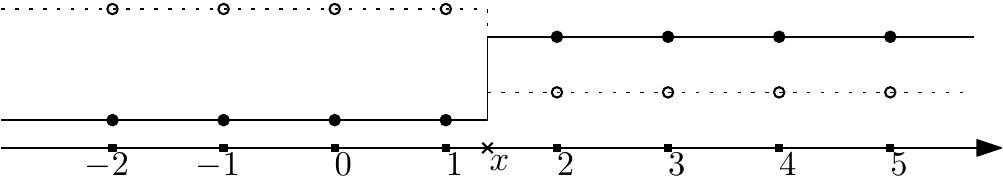}
\label{fig:ex6}
}
\caption{In all the pictures, the behavior of $\pp$ is pictured by solid lines (in dense time) and discs (in sampled discrete time); the behavior of $\qq$ is pictured by dotted lines (in dense time) and circles (in sampled discrete time); the higher value in any behavior corresponds to a $\logictrue$ truth value; for $i = 1, \ldots, 6$, $c_i$ denotes the dense-time behavior and $d_i$ its discrete-time sampling.}
\label{fig:all-examples}
\end{figure}

Before delving into the technical details of sampling invariance for generic MTL formulas, this sub-section illustrates the fundamental ideas that underlie the results of the paper.
The presentation is deliberately partly informal and based on examples, with the goal of stimulating the intuition that substantiates the choice of adaptations (in Section \ref{sec:canon-adapt}) and the rationale of the technical proofs (in Section \ref{th:sinv-flat}).
In all the examples of this sub-section, we assume a sampling period $\delta = 1$.

The first example demonstrates the need for non-Berkeley behaviors with the same $\delta$ as the chosen sampling period.
Consider formula $\diamondMTL{[2,5]}{\pp}$ and the behavior for $\pp$ in Figure \ref{fig:ex1}.
$\diamondMTL{[2,5]}{\pp}$ holds everywhere in dense time, but $\pp$ keeps on switching truth value in such a way that it is false at every sampled instant.
If the sampling period is not commensurate to the ``speed'' of the dense-time behavior there is always the possibility of similarly twisted behaviors which prevent achieving c.u.s.\ even for very simple formulas.
This justifies using the same $\delta$ for the non-Berkeley behaviors $\bsrd$ and the sampling function $\varsigma_{\delta, z}$.

If we assume such a constraint on the behaviors considered, c.u.s.\ is straightforward for ``existential'' --- that is ``eventually'' --- formulas.
Consider again formula $\diamondMTL{[2,5]}{\pp}$ and the behavior for $\pp$ in Figure \ref{fig:ex2}.
It should be clear that $c_2 \models \diamondMTL{[2,5]}{\pp}$ because $\pp$ holds at least once in any closed interval of length $3$.
It follows that the same holds for the discrete-time sampling $d_2$ of $c_2$.
In fact, consider any interval $I$ of size $3$ with integer endpoints and an instant within $I$ where $\pp$ holds.
Non-Berkeleyness entails that $\pp$ holds until the next sampling instant, since the previous sampling instant, or both.
Hence, it reaches a sampling instant that fits the interval $I$ over discrete time, which satisfies formula $\diamondMTL{[2,5]}{\pp}$ over discrete time.
This can be generalized to show that no change in the time interval is required for existential formulas when passing from dense- to discrete-time interpretations --- except for scaling the units according to the sampling period.
As a concrete example in Figure \ref{fig:ex2}, evaluate $\diamondMTL{[2,5]}{\pp}$ at $-1$, which references the dense-time interval $[1,4]$.
Consider the instant between $1$ and $2$ marked with a cross where $\pp$ holds; $\pp$ also holds since $1$, which is a sampling instant that belongs to the discrete-time interval $[1,4]$.

A similar reasoning works for ``universal'' --- that is ``always'' --- formulas, such as $\qq \Rightarrow \boxMTL{[2,5]}{\pp}$.
Behavior $c_3$ in Figure \ref{fig:ex3} is such that $c_3 \models \qq \Rightarrow \boxMTL{[2,5]}{\pp}$; in particular $\pp$ has to hold over the dense-time interval $[1, 4]$.
Over discrete-time, sampled values of $\pp$ have to hold over the \emph{discrete}-time interval with the same endpoints, which is obviously the case.
Again, this generalizes to universal formulas, which do not require changes in the time intervals when adapting them from dense- to discrete-time interpretations.

Things are more convoluted for c.u.i.s., which mandates changing the size of the intervals according to the type of formula --- existential or universal.
Let us consider again the existential formula $\diamondMTL{[2,5]}{\pp}$; it holds everywhere in the discrete-time behavior $d_4$ in Figure \ref{fig:ex4}.
If, however, the same formula is interpreted over the dense-time behavior $c_4$, of which $d_4$ is a sampling, it does not hold everywhere.
In particular, it holds at $-2$ and $-1$ but it does not hold in the open interval $(-2, -1)$: see the cross mark and the corresponding interval of size $3$ starting between $0$ and $1$.
The problem here is that non-Berkeleyness is a constraint on speed, not synchronization: the two samplings of $\pp$ at $0$ and $4$ record the value of the dense-time behavior $c_4$ respectively right before $0$ and right after $4$, hence leaving it unconstrained in the open interval $(0, 4)$ of size larger than $3$.
The ``interval of uncertainty'' is never larger than one sampling period on each side, hence we suggest to introduce an $\integers$-to-$\reals$ adaptation that grows intervals in existential formulas by this amount, thus accommodating the uncertainty in the worst case.
In the example, the adapted formula is $\diamondMTL{[1, 6]}{\pp}$ which clearly holds everywhere over $c_4$.

The dual reasoning suggests the adaptation for universal formulas such as $\qq \Rightarrow \boxMTL{[2,5]}{\pp}$.
In Figure \ref{fig:ex5}, the formula holds everywhere over discrete time. 
Over dense time, however, $\qq$ holds shortly before $-1$ (see cross mark) but $\pp$ does not hold everywhere in the corresponding interval of size $3$ starting shortly before $1$.
Again, a weaker formula holds over dense time, obtained by shrinking intervals in universal formulas by one sampling period on each side; the $\integers$-to-$\reals$ adaptation has to implement such a modification.
In the example, the adapted formula is $\qq \Rightarrow \boxMTL{[3,4]}{\pp}$ which clearly holds everywhere over $c_5$.

In order to rigorously extend the informal reasoning so far to arbitrary flat MTL formulas, we have to combine ``eventually'' and ``always'' formulas with the binary \emph{until} and \emph{release} modalities.
Let us demonstrate the intuition behind handling the former which turns out to be more intricate.
Consider a qualitative formula $\untilMTL{}{\qq, \pp}$ and the behavior in Figure \ref{fig:ex6}; let $x$ denote the time instant between $1$ and $2$ marked with a cross and assume that $\pp$ holds, in particular, precisely at $x$.
Then, the \emph{until} formula $\untilMTL{}{\qq, \pp}$ holds continuously over the interval $(-\infty, 1]$ in dense time (and beyond up to $x$).
Correspondingly, the same formula holds over the discrete interval $(-\infty, 1]$ in discrete time.
This suggests that \emph{until} formulas are c.u.s., as we will demonstrate formally in the rest of the paper.

Closure under inverse sampling is, again, more problematic.
Consider the same formula $\untilMTL{}{\qq, \pp}$ and the same discrete-time behavior $d_6$; we have seen that the \emph{until} formula holds over the discrete interval $(-\infty, 1]$ in discrete time.
Take a slightly different dense-time behavior, one where $\pp$ is false and $q$ is true at $x$ and everything else is as in $c_6$; let us name $c'_6$ this modified behavior.
Obviously, $d_6$ is a sampling of $c'_6$ as well as $c_6$.
However, $\untilMTL{}{\qq, \pp}$ does not hold anywhere in $(-\infty, 1]$ over $c'_6$ because $\pp$ becomes true left-continuously at $x$, which is incompatible with the dense-time semantics of \emph{until}.
In this case, the $\integers$-to-$\reals$ adaptation will have to replace the second argument $\pp$ of the \emph{until} formula with the weaker $\becfMTL{\pp}$ which holds at $x$ in $c'_6$ (as well as in $c_6$).
Alternatively, no adaptation is needed if we consider the stronger \emph{matching} variant of \emph{until} $\untilMMTL{\ }{\qq, \pp}$, where $\pp$ and $\qq$ would have to hold together at $1$ or $2$ in discrete time.

The following sub-sections present rigorous proofs of s.i.\ of MTL formulas that build upon the intuition behind the examples in the present sub-section.

\subsection{Sampling Invariance for MTL} \label{sec:sampl-invar-mtl}
This section provides a proof of the following fundamental result: there exist two suitable regular adaptations $\adapt{}, \unadapt{}$ such that \flatMTL{} is s.i.\ for the canonical sampling $\samp{}$.
In addition, the adaptations can be proved to introduce \emph{minimal} changes in the intervals of the adapted formulas, in the sense of Theorems \ref{th:optimal-canonical-1} and \ref{th:optimal-canonical-2} below.

\subsubsection{Canonical Adaptations}  \label{sec:canon-adapt}
Consider $\reals$-to-$\integers$ adaptation $\adapt{}$, parametric with respect to positive real parameter $\delta$, defined inductively as follows.
\begin{displaymath}
  \begin{array}{lcl}
  \adapt{\pi}  &  \triangleq  & \   \pi  \\

  \adapt{\untilMTL{\langle l, u \rangle}{\phi_1, \phi_2}}  &  \triangleq  &
        \  \untilMTL{[\lfloor l/\delta \rfloor, \lceil u/\delta \rceil ]}{\adapt{\phi_1}, \adapt{\phi_2}}  \\

  \adapt{\sinceMTL{\langle l, u \rangle}{\phi_1, \phi_2}}  &  \triangleq  &
        \  \sinceMTL{[\lfloor l/\delta \rfloor, \lceil u/\delta \rceil ]}{\adapt{\phi_1}, \adapt{\phi_2}}  \\

  \adapt{\relMTL{\langle l, u \rangle}{\phi_1, \phi_2}}  &  \triangleq  &
        \  \relMTL{\langle l', u' \rangle}{\adapt{\phi_1}, \adapt{\phi_2}} \\
	&  &  \text{where } l' = \begin{cases}
			                   \lfloor l/\delta \rfloor  &  \text{if } \langle \text{ is } (  \\
			                   \lceil l/\delta \rceil  &  \text{if } \langle \text{ is } [  \\
		                      \end{cases}  \\
	&  &  \text{and } u' = \begin{cases}
			                   \lceil u/\delta \rceil  &  \text{if } \rangle \text{ is } )  \\
			                   \lfloor u/\delta \rfloor  &  \text{if } \rangle \text{ is } ]  \\
		                      \end{cases}  \\

  \adapt{\redMTL{\langle l, u \rangle}{\phi_1, \phi_2}}  &  \triangleq  &
        \  \redMTL{\langle l', u' \rangle}{\adapt{\phi_1}, \adapt{\phi_2}} \\
	&  &  \text{where } l' = \begin{cases}
			                   \lfloor l/\delta \rfloor  &  \text{if } \langle \text{ is } (  \\
			                   \lceil l/\delta \rceil  &  \text{if } \langle \text{ is } [  \\
		                      \end{cases}  \\
	&  &  \text{and } u' = \begin{cases}
			                   \lceil u/\delta \rceil  &  \text{if } \rangle \text{ is } )  \\
			                   \lfloor u/\delta \rfloor  &  \text{if } \rangle \text{ is } ]  \\
		                      \end{cases}  \\

  \adapt{\phi_1 \wedge \phi_2}  &  \triangleq  & \   \adapt{\phi_1} \wedge \adapt{\phi_2} \\

  \adapt{\phi_1 \vee \phi_2}  &  \triangleq  & \  \adapt{\phi_1} \vee \adapt{\phi_2}
  \end{array}
\end{displaymath}

Consider $\integers$-to-$\reals$ adaptation $\unadapt{}$, parametric with respect to positive real parameter $\delta$, defined inductively as follows.\footnote{The restriction to closed intervals is clearly without loss of generality over discrete time.}

\begin{displaymath}
  \begin{array}{lcl}
  \unadapt{\pi}  &  \triangleq  &  \ \pi \\

  \unadapt{\untilMTL{[ l, u ]}{\phi_1, \phi_2}}  &  \triangleq  &
	   \ \untilMTL{( (l-2)\delta, (u+1)\delta )}{\unadapt{\phi_1}, \becfMTL{\unadapt{\phi_2}}} \\

  \unadapt{\sinceMTL{[ l, u ]}{\phi_1, \phi_2}}  &  \triangleq  &
  \ \sinceMTL{( (l-2)\delta, (u+1)\delta )}{}\!\left(\unadapt{\phi_1}, \becpMTL{\unadapt{\phi_2}}\right) \\

  \unadapt{\relMTL{[ l, u ]}{\phi_1, \phi_2}}  &  \triangleq  & 
	   \ \relMTL{[ (l+1)\delta, (u-1)\delta ]}{\unadapt{\phi_1}, \unadapt{\phi_2}}  \\

  \unadapt{\redMTL{[ l, u ]}{\phi_1, \phi_2}}  &  \triangleq  & 
	   \ \redMTL{[ (l+1)\delta, (u-1)\delta ]}{\unadapt{\phi_1}, \unadapt{\phi_2}}  \\

  \unadapt{\phi_1 \wedge \phi_2}  &  \triangleq  &    \unadapt{\phi_1} \wedge \unadapt{\phi_2} \\

  \unadapt{\phi_1 \vee \phi_2}  &  \triangleq  &   \unadapt{\phi_1} \vee \unadapt{\phi_2} \\
  \end{array}
\end{displaymath}

The proof of Theorem \ref{th:sinv-flat} will show that the asymmetry in the adaptation for \emph{until} (and \emph{since}) operators is needed to reconcile the non-matching semantics over discrete and dense time.
Alternatively, one can assume discrete-endpoint \emph{until} (and \emph{since}) operators in their matching variant (see Table \ref{tab:mtl-derived}) which preserves the symmetry in the adaptations.
We include them explicitly in the treatment also because they will be useful for the results of Section \ref{sec:verification}.

\begin{displaymath}
  \begin{array}{lcl}
  \unadapt{\untilMMTL{[ l, u ]}{\phi_1, \phi_2}}  &  \triangleq  &
	   \ \untilMMTL{( (l-1)\delta, (u+1)\delta )}{\unadapt{\phi_1}, \unadapt{\phi_2}} \\

  \unadapt{\sinceMMTL{[ l, u ]}{\phi_1, \phi_2}}  &  \triangleq  &
  \ \sinceMMTL{( (l-1)\delta, (u+1)\delta )}{}\!\left(\unadapt{\phi_1}, \unadapt{\phi_2}\right) \\
  \end{array}
\end{displaymath}

We name $\adapt{}$ and $\unadapt{}$ \emph{canonical} adaptations.

\subsubsection{Flat MTL is Sampling Invariant}
The main result of the paper is now proved.

\begin{theorem}[Sampling invariance of $\flatMTL{}$] \label{th:sinv-flat}
Let $\delta > 0$ be any sampling period and $z$ be any origin.
All flat $\flatMTL{}$ formulas are \emph{sampling invariant} with respect to the canonical adaptations $\adapt{}, \unadapt{}$ and the canonical sampling function $\samp{}$.
\end{theorem}

\begin{proof}
The proof is split into two parts: first we show that any dense-endpoint flat formula $\phi$ is c.u.s.; then we show that any discrete-endpoint flat formula $\phi$ is c.u.i.s.\footnote{For brevity we omit dealing with past operators, as it can be done from the corresponding future operators with little effort.}

Let us introduce the following abbreviations: for a dense-time instant $r$, let $\prevsamp{r}$ denote the sampling instant $z + \lfloor (r - z)/\delta \rfloor \delta$, which is immediately before or exactly at $r$, and let $\nextsamp{r}$ denote the sampling instant $z + \lceil (r - z)/\delta \rceil \delta$, which is immediately after or exactly at $r$.
Also, $\prevsampdist{r}$ and $\nextsampdist{r}$ denote the distances between $r$ and its previous and next sampling instant, respectively; that is $\prevsampdist{r} = r - \prevsamp{r}$ and $\nextsampdist{r} = \nextsamp{r} - r$.
Obviously $\prevsampdist{r}, \nextsampdist{r} \geq 0$.\footnote{This proof exploits some properties of the floor and ceiling functions. We refer the reader to \cite{GKP94} for a thorough treatment of these functions.}

\paragraph*{(Closure under sampling).}
Let $\phi$ be a generic dense-endpoint flat MTL formula, $b$ a dense-time non-Berkeley behavior in $\worldnb{\phi}{\reals}{\delta}$, and $\phi' = \adapt{\phi}$.
Then, let $b'$ be the sampling $\samp{b}$ of $b$ with the given origin and sampling period.

For a generic sampling instant $t = z + k\delta$, we show that $b(t) \modelstime{\reals} \phi$ implies $b'(k) \modelstime{\integers} \phi'$, by induction on the structure of $\phi$.
This proves that if $b \mdr \phi$ then $\samp{b} \mdz \adapt{\phi}$, for any $b \in \bsrd$; hence any $\flatMTL{}$ dense-endpoint formula is c.u.s.

\begin{itemize}
  \item $\phi = \pi$, $\phi = \phi_1 \wedge \phi_2$, and $\phi = \phi_1 \vee \phi_2$ are straightforward from the definitions.

  \item $\phi = \untilMTL{\langle l, u \rangle}{\pi_1, \pi_2}$.  
	 $\phi'$ is $\untilMTL{[l', u']}{\pi_1, \pi_2}$, with $l' = \lfloor l/\delta \rfloor$ and $u' = \lceil u/\delta \rceil$.

	 Let $d$ be a real in $\langle l, u \rangle$ such that $b(t+d) \mdr \pi_2$ and, for all $e \in [0, d)$, it is $b(t+e) \mdr \pi_1$.
    Since $b$ in non-Berkeley, there exists a $p \in [0, \delta]$ such that for all $f \in [-p, -p+\delta]$ it is $b(t+d+f) \mdr \pi_2$; i.e., $\pi_2$ holds over $I = [t+d-p, t+d-p+\delta]$.
    Some sampling instant must fall within $I$, as $I$ has size $\delta$.

	 In particular, it is either $p \geq \prevsampdist{t+d}$ or $-p + \delta \geq \nextsampdist{t+d}$: otherwise it would be $\delta = p + (-p + \delta) < \prevsampdist{t+d} + \nextsampdist{t+d} = \nextsamp{t+d} - \prevsamp{t+d} = \delta (\lceil (t+d-z)/\delta \rceil - \lfloor (t+d-z)/\delta \rfloor) \leq \delta$, a contradiction (where we exploited the property: $\lceil r \rceil - \lfloor r \rfloor \leq 1$ for any real $r$).
	 So, let $t'$ be the sampling instant:
	 \begin{displaymath}
		t' = \begin{cases}
		     \prevsamp{t+d} &  \text{if $p \geq \prevsampdist{t+d}$}  \\
		     \nextsamp{t+d} &  \text{otherwise}
		     \end{cases}
	 \end{displaymath}	 
	 It is not difficult to check that $(t' - t)/\delta \in [l', u']$.
	 In fact:
	 \begin{itemize}
		  \item if $p \geq \prevsampdist{t+d}$, then $t' - t = \prevsamp{t+d} - t = \delta (\lfloor (k\delta + d)/\delta\rfloor - k) = \delta \lfloor d/\delta \rfloor$.
			 Recall that $d \in \langle l, u \rangle$, and then \emph{a fortiori} $d \in [l, u] \supseteq \langle l, u \rangle$.
			 So $d/\delta \in [ l/\delta, u/\delta ]$, and $(t'-t)/\delta = \lfloor d/\delta \rfloor \in [\lfloor l/\delta \rfloor, \lfloor u/\delta \rfloor] \subseteq  [l', u']$.

		  \item if $p < \prevsampdist{t+d}$, then $t' - t = \nextsamp{t+d} - t = \delta (\lceil (k\delta + d)/\delta \rceil - k) = \delta \lceil d/\delta \rceil$.
			 Recall that $d \in \langle l, u \rangle$, and then \emph{a fortiori} $d \in [l, u] \supseteq \langle l, u \rangle$.
			 So $d/\delta \in [ l/\delta, u/\delta ]$, and $(t'-t)/\delta = \lceil d/\delta \rceil \in [\lceil l/\delta \rceil, \lceil u/\delta \rceil] \subseteq  [l', u']$.
	 \end{itemize}

In all, $b(t') \mdr \pi_2$.
By inductive hypothesis, it follows that for $d' = (t'-t)/\delta$ it is $b'(k+d') \mdz \pi_2$, and $d' \in [l', u']$.

Let us now show that for all integers $e' \in [0, d'-1]$ it is $b'(k+e') \mdz \pi_1$.
Recall that $d' \leq \lceil d/\delta \rceil < d/\delta + 1$, since $\lceil r \rceil < r + 1$ for any real number $r$; hence $\delta (d' - 1) < \delta (d/\delta) = d$.
Since for all $e \in [0, d)$ we have $b(t+e) \mdr \pi_1$, and since $[0, \delta (d' - 1)] \subset [0, d)$, \emph{a fortiori} for all $e \in [0, \delta (d'-1)]$ it is $b(t+e) \mdr \pi_1$.
By inductive hypothesis, it follows that for all integers $e' \in [0, d'-1] = [0, d')$ it is $b'(k+e') \mdz \pi_1$.
We conclude that $b'(k) \mdz \phi'$.

       \item $\phi = \relMTL{\langle l, u \rangle}{\pi_1, \pi_2}$.
	 $\phi'$ is $\relMTL{\langle l', u' \rangle}{\pi_1, \pi_2}$, where $l', u'$ depend on whether $I = \langle l, u \rangle$ is closed, open, or half-open.

        Let $d'$ be a generic integer in $\langle l', u' \rangle$; we show that $b'(k+d') \mdz \pi_2$ or there exists a $e' \in [0, d')$ such that $b'(k+e') \mdz \pi_1$.
	First we show that $\langle l', u' \rangle \subseteq \langle l/\delta, u/\delta \rangle$.
	In fact, consider the four possible cases for interval $I' = \langle l', u' \rangle$.
		\begin{itemize}
		  \item $I = [l, u]$, so $I' = [l', u']$, where $l' = \lceil l/\delta \rceil$ and $u' = \lfloor u/\delta \rfloor$.
		    Thus, $[l', u'] \subseteq [l/\delta, u/\delta]$, as $\lfloor r \rfloor \leq r$ and $\lceil r \rceil \geq r$ for any real $r$.

		  \item $I = [l, u)$, so $I' = [l', u')$, where $l' = \lceil l/\delta \rceil$ and $u' = \lceil u/\delta \rceil$.
		      Thus, $[l', u') \subseteq [l/\delta, u/\delta)$, as $[l', u') = [\lceil l/\delta \rceil, \lceil u/\delta \rceil - 1] \subseteq [l/\delta, u/\delta)$, noting that $\lceil r \rceil \geq r$, and that $\lceil r \rceil - 1 < r$, for any real $r$.

		  \item $I = (l, u]$, so $I' = (l', u']$, where $l' = \lfloor l/\delta \rfloor$ and $u' = \lfloor u/\delta \rfloor$.
			  Thus, $(l', u'] \subseteq (l/\delta, u/\delta]$, as $(l', u'] = [\lfloor l/\delta \rfloor + 1, \lfloor u/\delta \rfloor] \subseteq (l/\delta, u/\delta]$, noting that $\lfloor r \rfloor \leq r$, and that $\lfloor r \rfloor + 1 > r$, for any real $r$.

		  \item $I = (l, u)$, so $I' = (l', u')$, where $l' = \lfloor l/\delta \rfloor$ and $u' = \lceil u/\delta \rceil$.
			 Thus, $(l', u') \subseteq (l/\delta, u/\delta)$, as $(l', u') = [\lfloor l/\delta \rfloor + 1, \lceil u/\delta \rceil - 1] \subset (l/\delta, u/\delta)$, noting that $\lfloor r \rfloor + 1 > r$, and that $\lceil r \rceil - 1 < r$, for any real $r$.
		\end{itemize}

		In all, $b(t+\delta d') \mdr \pi_2$ or there exists a $e \in [0, \delta d')$ such that $b(t+e) \mdr \pi_1$.

	   If the former is the case, $b'(k+d') \mdz \pi_2$ holds by inductive hypothesis, which fulfills the goal.
	   If the latter is the case, we have $b(t) \mdr \diamondMTL{[0, \delta d')}{\pi_1} \equiv \untilMTL{[0, \delta d')}{\logictrue{}, \pi_1}$, which entails $b'(k) \mdz \diamondMTL{[0, d')}{\pi_1}$. 
	   Therefore, there exists a $e' \in [0, d')$ such that $b'(k+e') \mdz \pi_1$, as required.
\end{itemize}

\paragraph*{(Closure under inverse sampling).}
Let us first introduce the following terminology; for any dense-endpoint formula $\psi$:
\begin{itemize}
  \item if $b(t) \mdr \boxMTL{< \delta}{\psi}$ (resp.~$b(t) \mdr \boxPMTL{< \delta}{\psi}$), $\psi$ ``shifts to the right (s.t.r.) at $t$'' (resp.~``shifts to the left (s.t.l.) at $t$'');
  \item if $b(t) \mdr \untilMTL{= c}{}\!\left(\psi, \becfMTL{\neg \psi}\right)$ (resp.~$b(t) \mdr \sinceMTL{= c}{}\!\left(\psi, \becpMTL{\neg \psi}\right)$) for some $c \in (0, \delta)$, or $b(t) \mdr \psi \wedge \nowonMTL{\neg \psi}$ (resp.~$b(t) \mdr \psi \wedge \uptonowMTL{\neg \psi}$) and $c = 0$, $\psi$ ``turns false in the future (t.f.f.) at $\tff{t}{c}$'' (resp.~``turned false in the past (t.f.p.) at $\tfp{t}{c}$'').
\end{itemize}

Let $\phi$ be a generic discrete-endpoint flat MTL formula, $b$ a discrete-time behavior in $\world{\phi}{\integers}$, and $\phi' = \unadapt{\phi}$.
Then, let $b'$ be a dense-time non-Berkeley behavior in $\bsrd$ such that $\samp{b'} = b$ with the given origin and sampling period.

For a generic sampling instant $t = z + k\delta$, in the remainder we show that: (1) $b'(t) \mdr \phi'$; (2) $\phi'$ s.t.r.\ at $t$, or there exist $c \in [0,\delta)$ and $\varpi \in \PL{}$ such that $\phi'$ and $\varpi$ both t.f.f.\ at $\tff{t}{c}$, or $\phi'$ is false at $t+\delta$; and (3) either $\phi'$ s.t.l.\ at $t$ or there exist $c \in [0,\delta)$ and $\varpi \in \PL$ such that $\phi'$ and $\varpi$ both t.f.p.\ at $\tfp{t}{c}$.

From these three facts we can prove that $\phi$ is c.u.i.s.\ by showing that $b'(t) \mdr \phi'$ for all $t \in \reals$.
First, (1) shows this fact for all $t = z + k\delta$ for some integer $k$.
Then, let $t_n = t + \delta$ and show that $\phi'$ holds over the generic $\delta$-length closed real interval $[t, t_n]$.
If $\phi'$ s.t.r.\ at $t$ or it s.t.l.\ at $t_n$, we are done.
If $\phi'$ is false at $t+\delta = z + (k+1)\delta$ we have a contradiction which also closes the proof.
Otherwise, from (2) and (3) we assume that: (a) $\phi'$ t.f.f.\ at $\tff{t}{c_p}$ for some $c_p \in [0,\delta)$ with some $\varpi_p \in \PL$; and (b) $\phi'$ t.f.p.\ at $\tfp{t_n}{c_n}$ for some $c_n \in -[0,\delta)$ with some $\varpi_n \in \PL$.
Note that $|(t_n + c_n) - (t + c_p)| = |\delta + c_n - c_p| \leq \delta$; non-Berkeleyness of $b' \in \bsrd$ entails that either the two change points $t+c_p$ and $t_n+c_n$ coincide or $c_n = c_p = 0$.
In both cases Lemma \ref{lemma:changepoints} implies a contradiction which closes the whole proof.
We remark that the proofs go through also for intervals of temporal operators with negative endpoints, possibly with minimal adjustments that we do not discuss explicitly for the sake of brevity.

Finally, we prove (1), (2), and (3) by induction on the structure of $\phi$.
  \begin{itemize}
	 \item $\phi = \pi$. 
		(1) From the definition of $\samp{}$, it follows that $b'(t) = b(k)$.
		
		(2) Consider Lemma \ref{lemma:changepoints} at $t$: there exist $c_n \geq t$ such that $\pi$ t.f.f.\ at $\tff{t}{c_n}$ or it holds indefinitely in the future.
		If the latter is the case, $\pi$ obviously s.t.r.
      (3) is proved similarly as (2), with respect to the past.

    \item $\phi = \phi_1 \wedge \phi_2$ and $\phi = \phi_1 \vee \phi_2$ are straightforward from the definitions.

	 \item $\phi = \untilMTL{[l, u]}{\pi_1, \pi_2}$.
           $\phi'$ is $\untilMTL{( l', u' )}{\pi_1, \becfMTL{\pi_2}}$, with $l' = (l-2)\delta$ and $u' = (u+1)\delta$.

		(1) Let us start by proving $b'(t) \mdr \psi$ with $\psi = \untilMTL{[(l-1)\delta, u\delta]}{\pi_1, \pi_2}$.
		  This implies $b'(t) \mdr \phi'$, as $[(l-1)\delta, u\delta] \subset (l', u')$ hence $\psi$ is stronger than $\phi'$.
        Proving a stronger formula will be necessary in steps (2) and (3). \\
        Let $d \in [l, u]$ be the integer time instant such that $b(k+d) \mdz \pi_2$, which exists by hypothesis.
      The case $d = 0$ is trivial, hence let us consider $d > 0$.
      Still by hypothesis, for all integers $e \in [0, d) = [0,d-1]$ it is $b(k+e) \mdz \pi_1$.
      By inductive hypothesis, for all real $\delta$-multiples $e' \in [0, d\delta)$ it is $b'(t+e') \mdr \pi_1$.
      If $b'(t+d\delta) \mdr \pi_1$ as well, let $d' = d\delta$; otherwise $\pi_1$ t.f.f.\ at some $\tff{t'}{0}$ with $t+(d-1)\delta \leq t' \leq t+d\delta$ and let $d' = t'-t \geq 0$.
      Notice that $d' \in [(l-1)\delta, u\delta]$.
      Correspondingly, $\pi_1$ holds over $[0, d') \oplus t$.
      In addition, a little reasoning should convince us that Lemma \ref{lemma:changepoints} for $\pi_2$ at $t+d\delta$ --- also considering the fact that $\pi_1$ t.f.f.\ at $\tff{t'}{0}$ unless it holds at $t+d\delta$ --- implies that $\pi_2$ must hold over $(d', d\delta] \oplus t$; hence $b'(t+d') \mdr \becfMTL{\pi_2}$.

      (2) Let $s$ be any value in $(0, \delta)$.
      Let $c = d' - s$: notice that $c \in (l', u')$ because $c > d' - \delta \geq (l-1)\delta - \delta = l'$ and $c < d' \leq u\delta < u'$.
      Since $t + s + c = t + d'$ we have already shown that $b' (t+s+c) \mdr \becfMTL{\pi_2}$.
      Moreover, $[s, s+c) \subset [0, d')$ thus $b'(t+s+f) \mdr \pi_1$ holds \emph{a fortiori} for all $f \in [0, c)$.
      All this proves $\phi'$ s.t.r.

      (3) Let $f$ be any value in $-(0, \delta)$.
      For $d'' = d' - f$ we have $d'' \in [(l-1)\delta, (u+1)\delta)$ and $b'(t+f+d'') \mdr \becfMTL{\pi_2}$.
      Also, by inductive hypothesis either $\pi_1$ t.f.p.\ at $\tfp{t}{c}$ for some $c \in [0, \delta)$ or $\pi_1$ s.t.l.
      In the latter case, $\phi'$ s.t.l.\ as well; in the former case, $\phi'$ t.f.p. at $\tfp{t}{c}$ as well.

		\item $\phi = \relMTL{[l, u]}{\pi_1, \pi_2}$.
		  $\phi'$ is $\relMTL{[l', u']}{\pi_1, \pi_2}$, with $l' = (l+1)\delta$ and $u' = (u-1)\delta$.

		  (1) Let us start by proving $b'(t) \mdr \psi$ with $\psi = \relMTL{[l\delta, u\delta]}{\pi_1, \pi_2}$.
		  This implies $b'(t) \mdr \phi'$, as $[l\delta, u\delta] \supset [l', u']$ hence $\psi$ is stronger than $\phi'$.
        Proving a stronger formula will be necessary in steps (2) and (3). \\
		  Let $d'$ be any real value in $[l\delta, u\delta]$; we prove that $b'(t+d') \mdr \pi_2$ or $b'(t+e') \mdr \pi_1$ for some $e' \in [0, d')$.
		  We discuss two cases.
		  \begin{itemize}
		  \item If $t+d'$ is a sampling instant, $d = d'/\delta$ is an integer, and $d \in [l, u]$.
			 Also, by hypothesis, $b(k+d) \mdz \pi_2$ or $b(k+e) \mdz \pi_1$ for some integer $e \in [0, d-1]$.
			 In the former case, $b'(t+d') \mdr \pi_2$ follows by inductive hypothesis.
			 Otherwise, $b'(t+e') \mdr \pi_1$ for $e' = e\delta$ and $e' \in [0, d' - \delta] \subset [0, d')$, also by inductive hypothesis.

		  \item If $t+d'$ is not a sampling instant, let $p' = d' - \prevsampdist{t+d'}$ and $n' = d' + \nextsampdist{t+d'}$; these are both integer multiples of $\delta$.
			 Notice that $p' > d' - \delta \geq l\delta - \delta = (l-1)\delta$, and $n' < d' + \delta \leq u\delta + \delta = (u+1)\delta$.
			 Therefore, the two integers $p = p'/\delta$ and $n = n'/\delta$ are such that $p, n \in [l, u]$.
          Hence, from the hypothesis $b(k) \mdr \phi$, one of the following two cases holds.
			 \begin{itemize}
				\item $b(k+p) \mdz \pi_2$ and $b(k+n) \mdz \pi_2$, with $n = p+1$.
              By inductive hypothesis, $b'(t+p') \mdr \pi_2$ and $b'(t+n') \mdr \pi_2$ follow.
              Since $b' \in \bsrd$ is non-Berkeley by hypothesis, $\pi_2$ holds over the whole real interval $[t+p', t+n'] = [t+p', t+p'+\delta]$ as well.
              In particular, $b'(t+d') \mdr \pi_2$ for $d' \in [p', p'+\delta]$.

				\item $b(k+e) \mdz \pi_1$ for some integer $e \in [0, p-1]$ or $e \in [0, n-1]$.
              From $(p-1)\delta \leq (n-1)\delta < (d' + \delta) - \delta = d'$, it follows that $e' = \delta e \in [0, d')$.
              $b'(t+e') \mdr \pi_1$ holds by inductive hypothesis.
			 \end{itemize}
		  \end{itemize}
        In all, $b'(t) \mdr \psi$ is established.

		  (2) Let $f$ be any value in $(0, \delta)$ and $d''$ be any real value in $[l', u']$.
		  Since $d'' + f \in [l\delta, u\delta]$, we already proved that $b'(t + d'' + f) \mdr \pi_2$ or $b'(t+c) \mdr \pi_1$ for some $c \in [0, d'' + f)$. \\
        If, for all $d''$, the \emph{stronger} fact that $b'((t + f) + d'') \mdr \pi_2$ or $b'(t+c) \mdr \pi_1$ for some $c \in [f, d'' + f) \subset [0, d'' + f)$ holds, then we have proved that $\phi'$ s.t.r.\ at $t$ --- because $t+c = t+f+(c-f)$ and $c-f \in [0, d'')$. \\
        Otherwise, there is some $d''$ such that: (a) $b'((t+f) + d'') \mdr \neg \pi_2$; (b) $b'(t+c') = b'((t+f)+(c'-f)) \mdr \neg \pi_1$ for all $c' \in [f, d''+f)$; and (c) $b'(t+c) \mdr \pi_1$ for some $c \in [0, f)$.
        Let $v$ be the smallest instant in $[0,f)$ such that $\becfMTL{\neg \pi_1}$ holds at $t+v$; this exists because $b'$ is non-Zeno.
        Lemma \ref{lemma:changepoints} entails that $\pi_1$ holds over interval $[0, v) \oplus t$, and $\pi_1$ t.f.f.\ at $\tff{t}{v}$.
        Hence, it can be seen that $\phi'$ t.f.f.\ at $\tff{t}{v}$ as well.

		  (3) Let $s$ be any value in $-(0, \delta)$ and $d$ be any real value in $[l', u']$.
		  Since $s+d \in [l\delta, u\delta]$, we have already shown that $b'(t+(s+d)) \mdr \pi_2$ or $b'(t+e') \mdr \pi_1$ for some $e' \in [0, s+d)$.
        In both cases it follows that $\phi'$ s.t.l.\ at $t$, in particular as $e'' = e' - s$ with $e'' \in [-s, d) \subset [0, d)$ and $t+e' = (t+s)+e''$.

      \item $\phi = \untilMMTL{[l, u]}{\pi_1, \pi_2}$. \\
        Proof is all similar to the case of the ``standard'' \emph{until} with the simplification that matchingness allows us to establish the stronger $b'(t) \mdr \untilMMTL{[l\delta, u\delta]}{\pi_1, \pi_2}$ in part (1).
        \qedhere
      \end{itemize}
\end{proof}

\subsubsection{Canonical Adaptations are Optimal}
Let us provide some justification for the particular choice of canonical adaptations.
In principle, more complex transformations could be devised such that Theorem \ref{sec:sampl-invar-mtl} still holds.
However, we aimed at introducing adaptations that change the structure of the formulas as little as possible, such that the transformed formulas are ``essentially the same'' as the original formulas, except for some adjustments required to bridge the gaps in terms of time units and granularity (see Example \ref{ex:densevsdiscrete}).

In a nutshell, adaptations should preserve the propositional and modal structure of a formula as much as possible.
To formalize this intuition we introduce the notion of regularity.
\begin{definition}[Regularity of adaptations.]
Let $\phi_1, \phi_2$ any pair of MTL formulas and $\genop{}{}$ a modality.
An adaptation $\upsilon$ is:
\begin{itemize}
\item \emph{Compositional} if it satisfies $\upsilon[\phi_1 \wr \phi_2] \equiv \upsilon[\phi_1] \wr \upsilon[\phi_2]$ for any $\wr \in \{\wedge, \vee\}$.
\item \emph{Propositional-preserving} if $\upsilon[\pp] \equiv \pp$ for any $\pp \in \alphabet$.
\item $\genop{}{}$\emph{-modality-preserving} if, for any interval $I$, $\upsilon[\genop{I}{\phi_1, \phi_2}] \equiv \genop{I'}{\upsilon[\phi_1], \upsilon[\phi_2]}$.
\end{itemize}
An adaptation is $\genop{}{}$\emph{-regular} if it is compositional, propositional-preserving, and $\genop{}{}$-modality-preserving.
An adaptation is \emph{regular} when it is $\genop{}{}$\emph{-regular} for every modality $\genop{}{} \in \{\untilMTL{}{}, \sinceMTL{}{}, \relMTL{}{}, \redMTL{}{}\}$.
\end{definition}

Canonical adaptations $\adapt{}, \unadapt{}$ are regular for all modalities, with the exception of $\unadapt{}$ which is not $\untilMTL{}{}$-modality-preserving for the non-matching variant of the \emph{until} modality.
A $\untilMTL{}{}$-modality-preserving $\integers$-to-$\reals$ adaptation, however, would not achieve sampling invariance: as noted in Section \ref{sec:canon-adapt}, the matching semantics is the most natural choice to bridge the discrete- and dense-time semantics.
Furthermore, canonical adaptations are the ``best'' among all possible regular sam\-pling-invariant adaptations, in the sense that the adapted intervals are as constraining as possible.
This should be intuitively understandable already from the proof of Theorem \ref{th:sinv-flat}, which would not stand if we introduced any relaxation in adapted interval bounds.
More formally we have the following.
\begin{theorem}[Optimality of $\adapt{}$] \label{th:optimal-canonical-1}
Let $\upsilon^\reals$ be a regular $\reals$-to-$\integers$ adaptation such that \emph{any} flat dense-endpoint $\phi \in \flatMTL{}$ is c.u.s.\ with respect to it and $\samp{}$.
Then, $d \mdz \adapt{\phi}$ implies $d \mdz \upsilon^{\reals}[\phi]$ for any behavior $d \in \bsz$.
\end{theorem}
\begin{proof}[Proof]
The proof relies on techniques very similar to those of Theorem \ref{th:sinv-flat}, hence only a proof sketch is provided.

The proof goes by contradiction: let $\upsilon^{\reals}$ be a $\reals$-to-$\integers$ regular adaptation such that there exist $\phi \in \flatMTL$ and $d \in \bsz$ with $d \mdz \unadapt{\phi}$ but $d \not\mdz \upsilon^{\reals}[\phi]$.
Then, we build $c' \in \bsrd$ and $\zeta \in \flatMTL$ such that $c' \mdr \zeta$, $\samp{c'} \mdz \adapt{\zeta}$, but $\samp{c'} \not\mdz \upsilon^\reals[\zeta]$; hence $\zeta$ is not c.u.s.\ with respect to $\upsilon^{\reals}$ and $\samp{}$.

Let $k \in \integers$ be such that $d(k) \not\mdz \upsilon^{\reals}[\phi]$, while recall that $d(k) \mdz \adapt{\phi}$.
The propositional structure of $\adapt{\phi}$ and $\upsilon^{\reals}[\phi]$ is the same, since both adaptations are regular.
Then, by induction on the same propositional structure of $\adapt{\phi}$ and $\upsilon^{\reals}[\phi]$, one can show that there exists a \emph{modality} $\genop{}{} \in \{ \untilMTL{}{}, \relMTL{}{}, \sinceMTL{}{}, \redMTL{}{} \}$ such that $d(k) \mdz \adapt{\genop{I}{\pi_1, \pi_2}}$ and $d(k) \not\mdz \upsilon^{\reals}[\genop{I}{\pi_1, \pi_2}]$ for some $\pi_1, \pi_2 \in \PL$.
Let us write $\genop{J}{\beta_1, \beta_2}$ for $\adapt{\genop{I}{\pi_1, \pi_2}}$, and $\genop{K}{\gamma_1, \gamma_2}$ for $\upsilon^\reals[\genop{I}{\pi_1, \pi_2}]$.
The proof goes on by case discussion on the modality $\genop{}{}$; for brevity we just show the case $\genop{}{} = \untilMTL{}{}$, but the remaining cases can be handled all similarly.

Let $i$ be one of $1, 2$.
From the definition of $\adapt{}$ and the regularity of $\upsilon^\reals$, it is $\beta_i \equiv \gamma_i \equiv \pi_i$.
In all, there exists a $u \in J$ s.t.\ $d(k+u) \mdz \pi_2$ and $d(h) \mdz \pi_1$ for all $h \in [0,u) \oplus k$.
Since we are assuming that $d(k) \not\mdz \untilMTL{K}{\gamma_1, \gamma_2}$, it must be
$u \not\in K$, so either $K \subseteq (-\infty,u-1]$ or $K \subseteq [u+1, +\infty)$.
The remainder assumes $K \subseteq (-\infty,u-1]$ and $u > 1$; the other cases can be handled along the same lines and are omitted for brevity.

The next step builds a new formula $\varphi \triangleq \untilMTL{I}{\yy, \zz}$ with fresh propositional letters $\yy, \zz$; and a new discrete-time behavior $e$ over $\{\yy, \zz\}$ defined as follows: $\zz \in e(j)$ iff $j \geq k+u$, and $\yy \in e(j)$ for all $j \in \integers$.
It follows that $e(k) \mdz \untilMTL{J}{\yy, \zz}$ but $e(k) \not\mdz \untilMTL{K}{\yy, \zz}$ because $\max K < u$.
Also notice that $\adapt{\varphi} = \untilMTL{J}{\yy, \zz}$ and $\upsilon^{\reals}[\varphi] = \untilMTL{K}{\yy, \zz}$.
Take a $c$ built as follows: $\zz \in c(t)$ iff $t > z + (k+u-1)\delta$ and $\yy \in c(t)$ for all $t \in \reals$.
It should be clear that $c \in \bsrd$, $c \in \unsamp{e}$, and $c(z+k\delta) \mdr \untilMTL{I}{\yy, \zz} = \varphi$, 
because $\zz$ holds to the right of $z + (k+u-1)\delta$ but is false before and at it (over $z + \delta( k \oplus K)$).
In addition, one can see that $c(t) \mdr \varphi$ holds for all $t \geq z + k\delta$.

The last step is as follows: let us build a new non-Berkeley dense-time behavior $c' \in \bsrd$ over propositions in $\{\yy, \zz\} \cup \{\xx\}$, i.e., the same propositions as $c$ plus a fresh one denoted by $\xx$.
$c'|_{\alphabet}$ is identical to $c$, whereas $\xx \in c'(t)$ iff $c(t) \not\mdr \varphi$; hence in particular $c'(z+k\delta) \mdr \varphi \wedge \neg \xx$.
Notice that such $c'$ is non-Berkeley for $\delta$.
Finally, consider formula $\zeta \in \flatMTL$ defined as $\xx \vee \varphi$. 
Clearly, $c' \mdr \zeta$ is the case by construction; hence $\samp{c'} \mdz \adapt{\zeta}$ follows from Theorem \ref{th:sinv-flat}.
Also, $\samp{c'}(k) \not\mdz \upsilon^\reals[\varphi]$ as the truth of $\varphi$ does not depend on $\xx$; and $\samp{c'}(k) \not\mdz \upsilon^\reals [\xx]$ as the regularity of $\upsilon^\reals$ implies $\upsilon^\reals [\xx] = \xx$ and $c'(z+k\delta) \not\mdr \xx$.
In all we have $c' \mdr \zeta$, $\samp{c'} \mdz \adapt{\zeta}$, and $\samp{c'} \not\mdz \upsilon^\reals[\zeta]$.
Hence, c.u.s.\ does not hold for formula $\zeta$ with respect to adaptation $\upsilon$ and $\samp{}$, which is the desired contradiction.
\end{proof}

With a very similar approach the following theorem about $\unadapt{}$ adaptation can be proved.
\begin{theorem}[Optimality of regular $\unadapt{}$] \label{th:optimal-canonical-2}
Let $\upsilon^\integers$ be a regular $\integers$-to-$\reals$ adaptation for all modalities $\untilMMTL{\ }{}, \sinceMMTL{\ }{}, \relMTL{}{}, \redMTL{}{}$ such that \emph{any} discrete-endpoint $\phi \in \flatMTL{}$ using only modalities in $\{\untilMMTL{\ }{}, \sinceMMTL{\ }{}, \relMTL{}{}, \redMTL{}{}\}$ is c.u.i.s.\ with respect to it and $\samp{}$.
Then, $c \mdr \unadapt{\phi}$ implies $c \mdr \upsilon^{\integers}[\phi]$ for any behavior $c \in \bsrd$.
\end{theorem}

\subsection{Generalizations} \label{sec:generalizations}
Theorem \ref{th:sinv-flat} proved that \flatMTL{} is sampling invariant.
We claimed previously that \flatMTL{} is an MTL fragment of significant expressiveness; the specification examples in \cite{FPR08-FM08,FPR08-ICFEM08} demonstrate this in practice.
Nevertheless, we are still interested in investigating to what extent Theorem \ref{th:sinv-flat} can be generalized to larger classes of MTL formulas.
More precisely, given that Theorems \ref{th:optimal-canonical-1} and \ref{th:optimal-canonical-2} showed that canonical adaptations are optimal, we look for larger MTL fragments that are still sampling invariant with respect to \adapt{} and \unadapt{}.
Thus, henceforth sampling invariance will always implicitly refer to sampling invariance with respect to \adapt{} and \unadapt{}.

Let us start by illustrating the rather apparent fact that, for any sampling period $\delta$, there exist MTL formulas that are not s.i.\ with respect to $\delta$.
\begin{example}[A formula not c.u.s.] \label{ex:noncus}
For an arbitrary sampling period $\delta$, let us consider formula $\psi_\delta = \Som{\boxMTL{\leq \delta}{\pp}}$ and show that it is not c.u.s.\ with respect to $\delta$.
Consider any $c \in \bsr_\delta$ such that $\pp \in c(t)$ iff $t \in V$ for some interval $V$ such that $\delta < |V| < 2\delta$; clearly, $c \mdr \psi_\delta$.
However, for any $z$ such that $\delta + z + \delta\lceil (\inf V - z)/\delta \rceil > \sup V$, $\pp$ holds at one unique sampling instant over $\samp{c}$ (see Figure \ref{fig:noncus}).
Hence, $\samp{c} \not\mdz \adapt{\psi_\delta}$, where $\adapt{\psi_\delta}$ corresponds to $\Som{\boxMTL{\leq 1}{\pp}}$, because $\boxMTL{\leq 1}{\pp}$ requires $\pp$ to hold over \emph{two} adjacent time instants.
From the fact that the choice of origin $z$ is arbitrary in the definition of sampling invariance (Definition \ref{def:samplinginvariance}), it follows that $\psi_\delta$ is not c.u.s.
\begin{figure}[!h]
  \centering
  \includegraphics{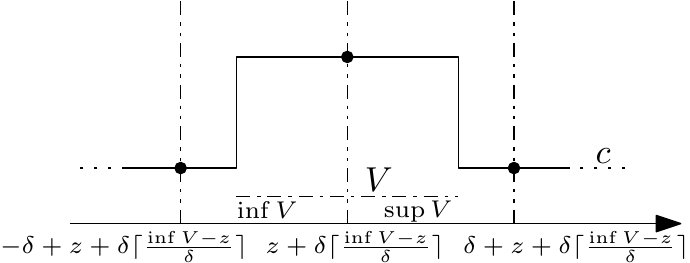}
  \caption{Behavior $c$ and its sampling $\samp{c}$.}
  \label{fig:noncus}
\end{figure}
\end{example}

\subsubsection{Shiftable Formulas}
Examples \ref{ex:noncus} and \ref{ex:nonflattable} suggest a straightforward criterion to identify non-flat MTL formulas that are c.u.s.: if non-Berkeleyness can be ``lifted'' from propositional letters to the truth value of some nested sub-formula $\lambda$, then the nesting formula containing $\lambda$ as a sub-formula can be flattened to one that is equi-satisfiable over non-Berkeley behaviors and does not introduce additional constraints.
To formalize this notion, we introduce the following.\footnote{This notion is very similar to the notion of stability introduced in \cite{Rab03}.}
\begin{definition}[$\epsilon$-shiftability] \label{def:shiftability}
Formula $\phi$ is $\epsilon$-shiftable, for some positive real $\epsilon$, iff $b_\phi \in \bsr_\epsilon$ holds for all $b \in \bsr_\epsilon$.
If $\phi$ is $\epsilon$-shiftable for \emph{any} $\epsilon$, it is called \emph{shiftable}.
\end{definition}

Shiftability provides a straightforward condition to determine larger MTL subsets that are c.u.s.\ and c.u.i.s., as the following theorem shows.\footnote{Recall the definition of \flatMTLplus{\Upsilon} at the end of Section \ref{sec:mtl-syntax}.}
\begin{theorem} \label{th:shiftable-to-si}
Let $\psi$ be a shiftable formula.
\begin{enumerate}
\item \label{part0} $\flatMTLplus{\{\psi\}}$ and $\flatMTL$ are equi-satisfiable over $\bsr_\delta$ for any $\delta$.
\item \label{part1} If $\psi$ and $\neg \psi$ are c.u.s., and $\neg \adapt{\psi} \equiv \adapt{\neg \psi}$ for all $\delta$, then all $\flatMTLplus{\{\psi\}}$ formulas are c.u.s.
\item \label{part2} If, for all $\psi' \in \unadaptinv{\psi}$ (where $\unadaptinv{}$ is the preimage of $\unadapt{}$), $\psi'$ and $\neg \psi'$ are c.u.i.s., and $\neg \unadapt{\psi'} \equiv \unadapt{\neg \psi'}$ for all $\delta$, then all $\flatMTLplus{\unadaptinv{\psi}}$ formulas are c.u.i.s.
\end{enumerate}
\end{theorem}
\begin{proof}
(\ref{part0}).
Every $\flatMTLplus{\{\psi\}}$ formula $\phi$ can be flattened into a $\flatMTL$ formula $\overline{\phi}$ by introducing an auxiliary propositional letter $\aaa$ that replaces every occurrence of $\psi$ and is declared to be logically equivalent to $\psi$ itself.
Since $\psi$ is shiftable, $b \in \bsr_\delta$ implies $b_{\psi}^{\psi \setminus \aaa} \in \bsr_\delta$; also, $b_{\psi}^{\psi \setminus \aaa} \in \bsr_\delta$ implies $b \in \bsr_\delta$ because $b$ has no more transition points than $b_{\psi}^{\psi \setminus \aaa}$.
Hence $\phi$ and $\overline{\phi}$ are equi-satisfiable.

(\ref{part1}).
Let $\phi$ be any formula in $\flatMTLplus{\{\psi\}}$, and consider a behavior $b \in \bst_\delta$ such that $b \mdr \phi$.
Let $\overline{\phi}$ denote the \emph{\flatMTL{}} formula obtained by replacing every occurrence of $\psi$ in $\phi$ by a fresh proposition $\aaa \in \alphabet$, and let $\phi^\circ$ be $\overline{\phi} \wedge \left( \aaa \Leftrightarrow \psi \right)$.
Clearly, $b_{\psi}^{\psi\setminus\aaa} \mdr \phi^\circ$, and $b_{\psi}^{\psi\setminus\aaa} \in \bst_\delta$ as we showed in (\ref{part0}).
Since $\overline{\phi}$ is flat, it is c.u.s.\ from Theorem \ref{th:sinv-flat}; hence $b' \mdz \adapt{\overline{\phi}}$ where $b' = \samp{b_{\psi}^{\psi\setminus\aaa}}$.
Notice that: $\adapt{\aaa \Leftrightarrow \psi} = \adapt{\aaa \wedge \psi \vee \neg \aaa \wedge \neg \psi}$ can be written as $\aaa \wedge \adapt{\psi} \vee \neg \aaa \wedge \neg \adapt{\psi}$.
From the c.u.s.\ of both $\psi$ and $\neg \psi$ and the fact that $\neg \adapt{\psi} \equiv \adapt{\neg \psi}$ we have $\adapt{\aaa \Leftrightarrow \psi} = \aaa \Leftrightarrow \adapt{\psi}$ and $b' \mdz \aaa \Leftrightarrow \adapt{\psi}$.
Let $\phi'$ be obtained from $\adapt{\overline{\phi}}$ by substituting every occurrence of $\aaa$ with $\adapt{\psi}$.
Hence, $b' \mdz \phi'$, which proves that $\phi$ is c.u.s.

(\ref{part2}).
All similar to (\ref{part1}), by noticing that $\unadapt{\psi'} = \psi$ for all $\psi' \in \unadaptinv{\psi}$ by definition of preimage.
\end{proof}

\subsubsection{LTL is Nestable}
Theorem \ref{th:shiftable-to-si} is applicable to a significant class of MTL formulas, namely qualitative formulas.
Indeed, LTL formulas are shiftable.\footnote{Note that non-strictness of LTL operators is necessary to have shiftability.}
\begin{lemma} \label{lem:ltl-shiftable}
All LTL formulas are shiftable.
\end{lemma}
\begin{proof}
Let us consider any non-Berkeley behavior $b \in \bst_\delta$ and any \LTL{} formula $\phi$.
By induction, we prove that $\tau(b_\phi) \subseteq \tau(b)$ which subsumes the lemma.

The base case $\phi = \pp$ is trivial.
The case $\phi = \neg \phi_1$ follows from the inductive hypothesis $\tau(b_{\phi_1}) \subseteq \tau(b)$ because $\tau(b_{\phi_1}) = \tau(b_{\neg \phi_1})$.

Let us consider $\phi = \untilMTL{}{\phi_1, \phi_2}$; we consider $b' = b_{\phi}|_\phi$ and prove that $\tau(b') \subseteq \tau(b)$.
To this end, let us first take any $t$ such that $b'(t) \mdr \phi$; hence $b(d) \mdr \phi_2$ for some $d \geq t$, and $b(u) \mdr \phi_1$ for all $u \in [t, d)$.
The semantics of the qualitative \emph{until} entails that $b'(t') \mdr \phi$ holds for all $t \leq t' \leq d$.
Then, $\phi$ cannot become false after $d$ until $\phi_2$ or $\phi_1$ becomes false; similarly, $\phi$ cannot become false before $t$ unless $\phi_1$ becomes false.
A dual argument shows that the same holds for $\neg \phi$.
This establishes that $\tau(b') \subseteq \tau(b)$.

The last case that has to be considered is $\phi = \phi_1 \wedge \phi_2$.
This is straightforward from the inductive hypothesis on $\phi_1$ and $\phi_2$: $\tau(b_{\phi_1}) \subseteq \tau(b)$ and $\tau(b_{\phi_2}) \subseteq \tau(b)$.
In addition, $\tau(b_{\phi_1 \wedge \phi_2}) \subseteq \tau(b_{\phi_1}) \cup \tau(b_{\phi_2})$ from the semantics of conjunction, hence $\tau(b_\phi) \subseteq \tau(b)$.
It is simple to check that $b_\phi \in \bst_\delta$ as well, because no left- and right- discontinuity can occur in $b_\phi$ as a result of applying conjunction.
\end{proof}

Based on the previous lemma, the following corollary of Theorem \ref{th:shiftable-to-si} shows that any LTL qualitative formula can be nested within $\flatMTL{}$ formulas without losing c.u.s.
\begin{corollary}
All $\flatMTLplus{\LTL}$ formulas are c.u.s.
\end{corollary}
\begin{proof}
The proof goes by induction on the nesting depth (i.e., the maximum number of nested modalities) of \LTL{} formulas.
For any integer $k > 0$, let $\LTL^k$ denote the set of all \LTL{} formulas of nesting depth $k$.

The base case is for any flat \LTL{} formula $\psi_1 \in \LTL^1$.
$\psi_1$ is shiftable from Lemma \ref{lem:ltl-shiftable}; $\psi_1$ and $\neg \psi_1$ are both c.u.s.\ from Theorem \ref{th:sinv-flat} (because $\neg \psi_1$ can also be written as a flat formula); one can check that $\neg \adapt{\psi_1} \equiv \adapt{\neg \psi_1}$ by pushing negations down to propositional letters.
So all $\flatMTLplus{\LTL^1}$ formulas are c.u.s.\ from Theorem \ref{th:shiftable-to-si}.

Let now $\psi_k \in \LTL^k$ be any \LTL{} formula of nesting depth $k > 1$.
$\psi_k$ is shiftable from Lemma \ref{lem:ltl-shiftable}; $\psi_k$ and $\neg \psi_k$ are both c.u.s., because they can both be written as $\flatMTLplus{\LTL^{k-1}}$ formulas, all of which are c.u.s.\ by inductive hypothesis; one can also check that $\neg \adapt{\psi_k} \equiv \adapt{\neg \psi_k}$ by pushing negations down to propositional letters and using the inductive hypothesis again.
So all $\flatMTLplus{\LTL^k}$ formulas are c.u.s.\ from Theorem \ref{th:shiftable-to-si}.
\end{proof}

A similar corollary for c.u.i.s.\ cannot be obtained along the same lines, due to the transformation of \emph{until} and its dual \emph{release} under the canonical adaptation $\unadapt{}$.


%% file: verification.tex
\section{Verification via Sampling} \label{sec:verification}
The notion of sampling invariance defines rigorously the connection between the non-Berkeley dense-time semantics and the discrete-time semantics of MTL, under the sampling relationship.
On the one hand, this allows the formal description --- by means of temporal logic formulas --- of systems where dense-time and discrete-time components evolve in parallel, and communicate through a sampler.
In addition, the theory of the previous sections can spawn several derived results that facilitate the analysis of real-time systems at the interface between discrete and dense time.
For instance, the notion of sampling can be used to describe system refinements from a ``physical'' dense-time model --- close to a ``real-world'' physical description --- to a more abstract discrete-time model --- which is implementable on digital hardware.

This section investigates another significant application of the notion of sampling and sampling invariance.
Namely, it builds a verification technique for dense-time MTL based on discretization.
The intuition is that, in order to analyze the behaviors induced by a set of dense-endpoint $\flatMTL$ formulas, their discrete-time samplings are analyzed instead.
The results about sampling invariance allow us to move the results of the discrete-time analysis back to the dense-time domain, under some restrictions.

The following Section \ref{sec:over-under-appr} shows how to build discrete-time under- and \oap{}s of any $\flatMTL$ formula.
The \oap{} embodies dis\-crete-time behaviors that are preserved into dense time, whereas the \uap{} represents discrete-time counter-examples that are preserved into dense time.
Together, they allow a partial reduction of dense-time satisfiability for $\flatMTL$ over non-Berkeley behaviors to dense-time MTL satisfiability.
In order to perform system verification --- i.e., checking if a given system satisfies certain putative properties --- the under- and \oap{}s of formulas can be combined to build two instances of the verification problem in the form of two validity checking problems for discrete-endpoint MTL formulas.
This procedure is shown in Section \ref{sec:system-verification}.
Finally, Section \ref{sec:examples-discussion} comments on a few key issues of this verification procedure, in particular its strengths and weaknesses from a mostly practical viewpoint.

\subsection{Under- and Over- Approximations} \label{sec:over-under-appr}
The over- and \uap{} functions $\underap{}, \overap{}$ are mappings from dense-endpoint $\flatMTL$ formulas to discrete-endpoint $\flatMTL$ formulas, parametric with respect to a sampling period $\delta$.
Given a $\flatMTL$ formula $\phi$, $\underap{\phi}$ and $\overap{\phi}$ retain some properties of the discrete-time \emph{samplings} of the dense-time behaviors in $\bsr_\delta$ satisfying $\phi$.
Correspondingly, it is possible to infer the validity of $\phi$ over dense time from the validity of its approximations.
For reasons that will become apparent shortly, $\underap{\phi}$ is named \emph{\uap{}} of $\phi$ and $\overap{\phi}$ \emph{\oap{}}.
Unsurprisingly, $\underap{}, \overap{}$ are closely related to canonical adaptations $\adapt{}, \unadapt{}$; in particular the \oap{} is a sort of inverse of the mapping $\unadapt{}$.
Their precise definition requires the introduction of the notion of granularity.

\subsubsection{Granularity} \label{sec:granularity}
For an MTL formula $\phi$, let $\Ical_{\phi} = \{ r_i / R_i \}_i$ be the set of all non-null, finite interval end-points appearing in $\phi$ and put in their irreducible form.\footnote{Recall that all finite endpoints are rationals (Section \ref{sec:mtl-syntax}).}
The \emph{granularity} $\rho_\phi$ of $\phi$ is defined as the pair:
  $\rho_\phi = (r_\phi, R_\phi) = \left( \gcd_{i} r_i ,  \lcm_{i} R_i \right)$.
Correspondingly, let us consider the set $\Dcal_\phi$ of rationals:\footnote{Recall that $a|b$ denotes that $b$ is an integer multiple of $a$.}
\begin{equation*}
  \Dcal_\phi \quad = \quad \left\{ \left. \frac{d}{D} \right| d | r_\phi \text{ and } R_\phi | D \right\}
\end{equation*}
It can be shown that, for any positive rational $\delta$ and $q \in \Ical_\phi$, $q/\delta$ is an integer iff $\delta \in \Dcal_\phi$; i.e., $\Dcal_\phi$ is the set of sampling periods $\delta$ such that any interval bound in $\phi$ is an integer when divided by $\delta$.
Notice that $\Dcal_\phi$ has a maximum (given by $r_\phi/R_\phi$) but no minimum.
Finally, for a set of formulas $\Phi$, $\Dcal_\Phi$ is defined as $\Dcal_{\widehat{\phi}}$ where $\widehat{\phi} \triangleq \bigwedge_{\varphi \in \Phi} \varphi$.

\subsubsection{Under-Approximation} \label{sec:underap}
The \uap{} function $\underap{}$ maps dense-endpoint MTL formulas to dis\-crete-endpoint MTL formulas such that the non-validity of the latter implies the non-validity of the former, over behaviors in $\bst_\delta$.
More precisely, $\underap{\phi}$ is defined only for MTL formulas such that $\delta$ is in $\Dcal_\phi$, where it coincides with $\adapt{\phi}$. \\
\begin{equation*}
  \begin{array}{lcl}
  \underap{\pi}  &  \triangleq  & \   \pi  \\

  \underap{\untilMTL{\langle l, u \rangle}{\phi_1, \phi_2}}  &  \triangleq  &
        \  \untilMTL{[l/\delta, u/\delta ]}{\underap{\phi_1}, \underap{\phi_2}}  \\

  \underap{\sinceMTL{\langle l, u \rangle}{\phi_1, \phi_2}}  &  \triangleq  &
        \  \sinceMTL{[l/\delta, u/\delta]}{\underap{\phi_1}, \underap{\phi_2}}  \\

  \underap{\relMTL{\langle l, u \rangle}{\phi_1, \phi_2}}  &  \triangleq  &
        \  \relMTL{\langle l/\delta , u/\delta \rangle}{\underap{\phi_1}, \underap{\phi_2}} \\

  \underap{\redMTL{\langle l, u \rangle}{\phi_1, \phi_2}}  &  \triangleq  &
        \  \redMTL{\langle l/\delta , u/\delta \rangle}{\underap{\phi_1}, \underap{\phi_2}} \\

  \underap{\phi_1 \wedge \phi_2}  &  \triangleq  & \   \underap{\phi_1} \wedge \underap{\phi_2} \\

  \underap{\phi_1 \vee \phi_2}  &  \triangleq  & \  \underap{\phi_1} \vee \underap{\phi_2}
  \end{array}
\end{equation*}

The following lemma justifies the name \emph{\uap{}}.
\begin{lemma}[Under-approximation] \label{lem:under}
For any dense-endpoint $\flatMTL{}$ formula $\phi$, $\delta \in \Dcal_\phi$, and $b \in \bsz$: if $b \not\mdz \underap{\phi}$ then for all $b' \in \bsr_\delta$ such that $\samp{b'} = b$ it is $b' \not\mdr \phi$.
\end{lemma}
\begin{proof}
$\phi$ is a dense-endpoint $\flatMTL$ formula, hence it is c.u.s.\ from Theorem \ref{th:sinv-flat}: for any $b \in \bsr_\delta$, if $b \mdr \phi$ then $\samp{b} \mdz \adapt{\phi}$.
By taking the contrapositive, and by noticing that $\adapt{}$ and $\underap{}$ coincide when they are both defined, we have that for any $b \in \bsr_\delta$, if $\samp{b} \not\mdz \underap{\phi}$ then $b \not\mdr \phi$.
\end{proof}

\subsubsection{Over-Approximation} \label{sec:overap}
The \oap{} function $\overap{}$ maps dense-endpoint MTL formulas to dis\-crete-endpoint MTL formulas such that the validity of the latter implies the validity of the former, over behaviors in $\bst_\delta$.
More precisely, $\overap{\phi}$ is defined only for MTL formulas such that $\delta$ is in $\Dcal_\phi$, where it is a pseudo-inverse of $\unadapt{}$. \\
\begin{equation*}
  \begin{array}{lcl}
  \overap{\pi}  &  \triangleq  & \   \pi  \\

  \overap{\untilMTL{\langle l, u \rangle}{\phi_1, \phi_2}}
                            &  \triangleq  & \  \untilMMTL{[l/\delta + 1, u/\delta - 1]}{\overap{\phi_1}, \overap{\phi_2}} \\

  \overap{\sinceMTL{\langle l, u \rangle}{\phi_1, \phi_2}}
                            &  \triangleq  & \  \sinceMMTL{[l/\delta + 1, u/\delta - 1]}{\overap{\phi_1}, \overap{\phi_2}} \\

  \overap{\relMTL{\langle l, u \rangle}{\phi_1, \phi_2}}
                            &  \triangleq  & \  \relMTL{[l/\delta - 1, u/\delta + 1]}{\overap{\phi_1}, \overap{\phi_2}} \\

  \overap{\redMTL{\langle l, u \rangle}{\phi_1, \phi_2}}
                            &  \triangleq  & \  \redMTL{[l/\delta - 1, u/\delta + 1]}{\overap{\phi_1}, \overap{\phi_2}} \\

  \overap{\phi_1 \wedge \phi_2}  &  \triangleq  & \   \overap{\phi_1} \wedge \overap{\phi_2}  \\

  \overap{\phi_1 \vee \phi_2}  &  \triangleq  & \   \overap{\phi_1} \vee \overap{\phi_2}
  \end{array}
\end{equation*}

The following lemma justifies the name \emph{\oap{}}.
\begin{lemma}[Over-approximation] \label{lem:over}
For any dense-endpoint $\flatMTL{}$ formula $\phi$, $\delta \in \Dcal_\phi$, and $b \in \bsz$: if $b \mdz \overap{\phi}$ then for all $b' \in \bsr_\delta$ such that $\samp{b'} = b$ it is $b' \mdr \phi$.
\end{lemma}
\begin{proof}
If $\phi$ is a dense-endpoint $\flatMTL$ formula, then $\overap{\phi}$ is a discrete-endpoint $\flatMTL$ formula.
Hence the latter is c.u.i.s.\ from Theorem \ref{th:sinv-flat}: for any $b \in \bsz$, if $b \mdz \overap{\phi}$ then $b' \mdr \unadapt{\overap{\phi}}$ holds for all $b' \in \bsr_\delta$ such that $b = \samp{b'}$.

One can check that the dense-time validity of the formula $\unadapt{\overap{\phi}} \Rightarrow \phi$ is guaranteed by the definitions of $\unadapt{}$ and $\overap{}$.
In particular, $\unadapt{} \circ \overap{}$ is an identity for \emph{release} (and \emph{trigger}) operators with closed intervals. 
On the other hand, $\unadapt{} \circ \overap{}$ yields stronger formulas for \emph{release} (and \emph{trigger}) operators with open intervals and for \emph{until} (and \emph{since}) operators.
The latter holds also from the fact that $\unadapt{\untilMMTL{[l,u]}{\pi_1, \pi_2}}$ is $\untilMMTL{((l-1)\delta, (u+1)\delta)}{\pi_1, \pi_2}$.
It is easy to check that these properties of basic operators can be lifted to whole formulas by application of straightforward propositional identities on the negation normal form in which MTL formulas are expressed.
In all, $b' \mdr \unadapt{\overap{\phi}}$ implies $b' \mdr \phi$.
\end{proof}

\subsection{MTL Verification} \label{sec:system-verification}
In the formal timed setting, verification consists in checking whether all behaviors generated by a system model (usually called \emph{specification}) satisfy some given putative property (usually called \emph{requirements}) \cite{HM96}.
Assume that both the specification and the requirements are formalized as MTL formulas $\system{}$ and $\prop$, respectively.
Verification of $\system$ against $\prop$ is equivalent to checking the validity of the dense-endpoint MTL formula $\verif = \Alw{\system} \Rightarrow \Alw{\prop}$.
If $\verif$ is valid, any behavior of the system also respects the requirements; i.e., we have checked that $\world{\system}{\reals} \subseteq \world{\prop}{\reals}$.
On the contrary, if $\verif$ is not valid, there exists at least one behavior of the system that violates the requirements; i.e., $\world{\system}{\reals} \cap \world{\neg \prop}{\reals}$ is not empty so $\world{\system}{\reals} \not\subseteq \world{\prop}{\reals}$.

In this section, we describe a verification algorithm that is applicable to specifications and requirements in $\flatMTL$ over non-Berkeley dense-time behaviors.

The algorithm is based on the following.
\begin{proposition}[Model approximations] \label{prop:approx}
For any $\flatMTL$ formulas $\phi_1, \phi_2$, and for any $\delta \in \Dcal_{\{ \phi_1, \phi_2 \}}$:
 \begin{enumerate}
	 \item \label{it1} if $\Alw{{\underap{\phi_1}}} \Rightarrow \Alw{\overap{\phi_2}}$ is $\integers$-valid,
		    then $\Alw{\phi_1} \Rightarrow \Alw{\phi_2}$ is $\reals^\delta$-valid;
	 \item \label{it2} if $\Alw{\overap{\phi_1}} \Rightarrow \Alw{\underap{\phi_2}}$ is not $\integers$-valid,
		    then $\Alw{\phi_1} \Rightarrow \Alw{\phi_2}$ is not $\reals^\delta$-valid.
 \end{enumerate}
\end{proposition}
\begin{proof}
(\ref{it1}).
Let $\delta \in \Dcal_{\{\phi_1, \phi_2\}}$.
Assume that $\Alw{{\underap{\phi_1}}} \Rightarrow \Alw{\overap{\phi_2}}$ is $\integers$-valid.
That is, for all $b \in \bsz$ it is $b \not \mdz \underap{\phi_1}$ or $b \mdz \overap{\phi_2}$.
From Lemmas \ref{lem:over} and \ref{lem:under}, this implies that for all $b \in \bsz$, for all $b' \in \bsrd$ such that $\samp{b'} = b$, it is either $b' \not \mdr \phi_1$ or $b' \mdr \phi_2$.
Since $\samp{}$ is total, for any $b' \in \bsrd$ there exists a $b \in \bsz$ such that $\samp{b'} = b$.
We conclude that for all $b' \in \bsrd$, either $b' \not \mdr \phi_1$ or $b' \mdr \phi_2$; i.e., $\Alw{\phi_1} \Rightarrow \Alw{\phi_2}$ is $\reals^\delta$-valid.

Proof of (\ref{it2}) is obtainable from the proof of (\ref{it1}) by duality.
\end{proof}

\subsubsection{Verification Algorithm}
Proposition \ref{prop:approx} suggests to introduce the following notation.
Given a set of formulas $\Sys = \{ \system^i \}_i$ such that $\system = \bigwedge_i \Alw{\system^i}$ represents a formal model of the system, and a formula $\prop$ that represents a formal statement of the requirements, let us define the discrete-endpoint formulas:
\begin{gather*}
  \overmodel{\phi} \quad \triangleq \quad \bigwedge_i \Alw{\underap{\system^i}} \ \Rightarrow \ \Alw{\overap{\prop}} \\
  \undermodel{\phi} \quad \triangleq \quad \bigwedge_i \Alw{\overap{\system^i}} \ \Rightarrow \ \Alw{\underap{\prop}}
\end{gather*}
Let us call $\overmodel{\phi}$ and $\undermodel{\phi}$ \emph{over-model} and \emph{under-model} of the system, respectively, (in analogy with Lemmas \ref{lem:over} and \ref{lem:under}) because the former preserves validity and the latter non-validity.

A verification algorithm for systems and properties specified as dense-endpoint $\flatMTL$ formulas can be formalized as follows, where $\proc{Z-valid?}$ is a validity-checking procedure for discrete-endpoint MTL formulas.
\begin{codebox}
\Procname{$\proc{\flatMTL{}-verify}(\delta: \reals_{> 0}, \id{\Sys = \{\system^i \}_i}, \id{\prop}: \flatMTL): \{\logictrue, \logicfalse, \const{fail}\}$}
\li \kw{assume} $\delta \in \Dcal_{\Sys \cup \{\prop\}}$
\li $\overmodel{\phi} \gets \bigwedge_i \Alw{\underap{\system^i}} \Rightarrow \Alw{\overap{\prop}}$
\li $\undermodel{\phi} \gets \bigwedge_i \Alw{\overap{\system^i}} \Rightarrow \Alw{\underap{\prop}}$
\li \If $\proc{Z-valid?}(\overmodel{\phi})$ \>\>\>\>\>\>\>\>\Comment $\overmodel{\phi}$ valid over discrete time?
\li \Then \Return $\logictrue$              \>\>\>\>\>\>\Comment verification over $\bst_\delta$ successful
\li \Else \If $\neg \proc{Z-valid?}(\undermodel{\phi})$ \>\>\>\>\>\>\Comment $\undermodel{\phi}$ not valid over discrete time?
\li       \Then \Return $\logicfalse$                   \>\>\>\>\Comment verification over $\bst_\delta$ not successful
\li       \Else \Return $\const{fail}$ \End \End        \>\>\>\>\Comment cannot conclude any verification result
\end{codebox}

The correctness of the algorithm follows directly from Proposition \ref{prop:approx}, keeping in mind that $b \mdt \Alw{\psi_1} \wedge \Alw{\psi_2}$ iff $b \mdt \Alw{\psi_1}$ and $b \mdt \Alw{\psi_2}$.

\subsubsection{Incompleteness}
A verification algorithm is \emph{complete} if, for any input, it terminates with a conclusive result about whether the given requirements $\prop$ are indeed a property of the system $\system$ or not.

The verification algorithm for $\flatMTL$ we provided above is \emph{incomplete}, as it can fail to provide a conclusive answer about whether $\prop$ is indeed a property of all behaviors of the system $\system$.
The incompleteness is two-fold.
First, the algorithm does not consider \emph{all} dense-time behaviors $\bsr$, but only those in $\bsrd$, i.e., ``slow'' with respect to some chosen sampling period $\delta$.
Hence, it may be that $\prop$ does not hold for some ``real'' behavior of the system which is ``fast'', i.e., for some behavior in $\world{\system}{\reals} \setminus \worldnb{\system}{\reals}{\delta}$.
Second, the under- and over-model $\undermodel{\phi}, \overmodel{\phi}$ are in general non-equivalent discrete-endpoint formulas.
Hence, it is possible that $\overmodel{\phi}$ is not valid and $\undermodel{\phi}$ is valid; if this is the case no conclusion about the verification of the system can be drawn.

Since the algorithm is parametric with respect to $\delta$, smaller values of $\delta$ can be tried in order to avoid the incompleteness hurdle.
Changing the value of $\delta$ affects the verification problem in two ways: more (``faster'') behaviors are considered for verification, and new under- and over-models are generated that represent a ``finer-grain'' discretization of the original problem.
These two aspects interact in subtle ways because they change the verification problem from two opposite sides.
By combining them, one may expect to achieve at least the following partial notion of completeness: if $\prop$ is a property of $\system$ over behaviors in $\bsrd$ for \emph{some} choice of $\delta$, then there exists a suitable choice of $\delta$ such that $\proc{\flatMTL{}-verify}(\delta, \id{\system}, \id{\prop})$ returns $\logictrue$; and conversely when $\prop$ is not a property of $\system$.
Unfortunately, the following example shows that even this weaker notion of completeness is not achieved by the algorithm.

\begin{example}[Incompleteness of the algorithm] \label{ex:incompleteness}
Consider a simple set of behaviors completely described by formulas in Table \ref{tab:system}.
It should be clear that all behaviors $b \in \world{\Alw{\system^1} \wedge \Alw{\system^2}}{\reals}$ of the system are such that $\pp$ holds on some interval $V = \langle t, +\infty)$ and $\neg \pp$ holds on the complement interval $\reals \setminus V$ (which is unbounded to the left).
Hence, any such $b$ satisfies property $\prop$ and is in $\bsrd$ for any $\delta$.
\begin{table}[!h]
\begin{center}
\begin{tabular}{c c c}
  $\system^1$  &$\triangleq$&  $\Som{\pp} \wedge \Som{\neg \pp}$   \\
  $\system^2$  &$\triangleq$&  $\pp \Rightarrow \boxMTL{}{\pp}$   \\
  $\prop$      &$\triangleq$&  $\pp \Rightarrow \diamondMTL{=1}{\pp}$
\end{tabular}
\caption{$\Sys$ and $\prop$.}
\label{tab:system}
\end{center}
\end{table}

Table \ref{tab:approximations} shows the over- and under-models of this system for any $\delta \in \linebreak \Dcal_{\Sys \cup \{\prop\}} = \{ 1/k \mid k \in \naturals_{> 0} \}$, after some simplifications (in particular $\overap{\system^2} \linebreak = \neg \pp \vee \boxMTL{[-1, +\infty]}{\pp}$ is equivalent to the formula in Table \ref{tab:approximations} under the global satisfiability semantics).
It is simple to check that, for any value of $\delta$, the over-model $\overmodel{\phi}$ is not valid because $\Alw{\overap{\prop}}$ contradicts $\Alw{\underap{\system^1}}$.
Also for any value of $\delta$ the under-model $\undermodel{\phi}$ is vacuously valid because $\Alw{\overap{\system^1}}$ is inconsistent with $\Alw{\overap{\system^2}}$.
In all, we cannot verify our system with our algorithm, no matter what value of sampling period we choose.
\begin{table}[htb]
\begin{center}
\begin{tabular}{ll}
$\underap{\system^1} = \Som{\pp} \wedge \Som{\neg \pp}$  &    $\overap{\system^1} = \Som{\pp} \wedge \Som{\neg \pp}$ \\
$\underap{\system^2} = \pp \Rightarrow \boxMTL{}{\pp}$   &  $\overap{\system^2} = \Alw{\pp} \vee \Alw{\neg \pp}$ \\
$\underap{\prop} = \pp \Rightarrow \diamondMTL{= k}{\pp}$  &  $\overap{\prop} = \neg \pp$
\end{tabular}
\caption{Under- and over-models of $\Sys, \prop$ for $\delta = 1/k$.}
\label{tab:approximations}
\end{center}
\end{table}
\end{example}

In spite of its incompleteness, in the next section we discuss why the verification algorithm can still provide practically very useful results.

\subsection{Discussion} \label{sec:examples-discussion}
In related work, we proved that MTL is fully decidable over dense-time non-Berkeley behaviors $\bsrd$ for any $\delta$ \cite{FR08-FORMATS08}, with the same worst-case complexity as discrete-time MTL; hence an \emph{incomplete} decision procedure may seem impractical.
In this section we demonstrate that this is not the case, and we discuss how the impact of incompleteness can be limited in practice with the application of a few good practices.

First of all, the decision procedure for MTL over $\bsrd$ --- the only one currently available \cite{FR08-FORMATS08} --- relies on a rather exotic decision procedure, which translates MTL to a family of uncommon decidable real-time temporal logics introduced by Hirshfeld and Rabinovich \cite{HR04}.
The decision procedures for such logics have never been implemented, and seem quite complex in practice.
More generally, the practical high complexity of deciding temporal logics over dense-time domains is witnessed not only by theoretical results, but also by the current scarcity of state-of-the-art tools that implement such decision procedures.
Even the well-known real-time temporal logic MITL, whose decidability over dense time is known since the seminal work of Alur, Feder, and Henzinger \cite{AFH96}, still lacks an implementation, despite the recent efforts towards simplifying its decision procedure \cite{HR05,MNP06}.

Compare this unsatisfactory picture to the vastly different scenario of (real-time) temporal logics over discrete time, where a significant number of off-the-shelf efficient verification tools are available (e.g., \cite{PMS07,BMPSS07,PSSM03,NuSMV2,Alaska} just to mention a few for LTL/MTL).
This suggests that a dense-time verification procedure based on discretization is very appealing from a practical viewpoint, because it can be implemented easily and it can rely on solid and scalable implementations.
In fact, in related work \cite{FPR08-FM08,FPR08-ICFEM08,BFPR09-SEFM09} we presented the straightforward implementation of the verification procedure described in this section, and we demonstrated its practical efficiency with a few non-trivial verification examples.

The same examples also show that the flat fragment of MTL retains (under the global satisfiability semantics) a significant expressive power, suitable to formalize typical behaviors of real-time systems.
For example, it is possible to describe runs of arbitrary timed automata or bounded time Petri nets over non-Berkeley behaviors.
The formalization in flat MTL of these complex abstract machines is far from straightforward and requires a careful analysis to avoid inconsistencies.
However, the experience of \cite{FPR08-FM08,FPR08-ICFEM08,BFPR09-SEFM09} can be leveraged and extended to similar systems described by means of the notions of state and transition.

Even the \emph{incompleteness} of our verification algorithm turns out not to be too large a handicap in practice.
More precisely, the fact that equivalent dense-endpoint formulas can yield nonequivalent discrete-time under- or over-approximations can be turned into an advantage: with some additional effort in writing the dense-time model of our system, we can often express it in a form whose over- and under-models are unaffected by incompleteness.
This effort can in general be non-trivial, but it can give very good practical results nonetheless.
The following example provides a few in-the-small demonstrations of our claims, whereas more complex cases have been introduced elsewhere \cite{FPR08-ICFEM08,BFPR09-SEFM09}.

\begin{example}
Let us go back to Example \ref{ex:incompleteness} and change formula $\system^2$ into $\systembis^2 \triangleq \pp \Rightarrow \boxMTL{\geq \delta}{\pp}$, according to the chosen sampling period $\delta$.
A little reasoning should convince us that $\Alw{\system^2}$ is equivalent to $\Alw{\systembis^2}$ over behaviors in $\bsrd$: if $\pp$ holds at some time $t$ as well as over the left-closed interval $t \oplus [\delta, +\infty)$, it cannot be false anywhere in $(t, t+\delta)$ because this would violate the hypothesis of non-Berkeleyness for the given $\delta$.
Let us take our system model to be $\Sys = \{\system^1, \systembis^2\}$, and let us build its over-model $\overmodel{\Sys}$.
Notice that $\overap{\systembis^2}$ can be computed as $\pp \Rightarrow \boxMTL{}{\pp}$; unlike $\overap{\system^2}$, this is an accurate discrete-time rendition of the dense-time model.
It is now possible to prove that $\overmodel{\phi}$ is $\integers$-valid for any $\delta = 1/k$, which verifies our system over dense time.

Let us now turn our attention to property $\prop$ in Example \ref{ex:incompleteness}.
It should be apparent that its over-approximation $\overap{\prop} = \neg \pp$ is very unsatisfactory, and it is unlikely to yield valuable results when used in an under-model.
Consider however formula $\prop' \triangleq \pp \Rightarrow \boxMTL{=1}{\pp}$; $\prop'$ is trivially equivalent to $\prop$.
However, its over-approximation is the much more reasonable $\pp \Rightarrow \boxMTL{[k-1, k+1]}{\pp}$ which is non-trivially satisfiable for any $k > 1$.
\end{example}


%% file: related.tex
\section{Related Work} \label{sec:related}
The relationship between dense and discrete real-time semantics has been investigated by many authors.
In this section we mention the approaches that are closest to ours, and we detail the most significant differences and relative merits.

The seminal paper by Henzinger, Manna, and Pnueli \cite{HMP92} is both the first and the best-known work dealing with the theme of dense vs.~discrete real-time through the notion of \emph{digitization}.
Given the significance of this notion, Section \ref{sec:comp-with-digit} is devoted to a detailed summary of it, as well as to a comparison with sampling invariance.
Section \ref{sec:other-related} succinctly describes other related work about the relation between dense and discrete time models for real-time formalisms.
Finally, briefly widening the scope beyond real-time notations, the results of this paper seem to bear a connection with the classical theory of digital sampling (e.g., \cite{sampling-book}).
Section \ref{sec:sampling-theorem} sketches a partly formal analysis of this alleged link.

\subsection{Comparison with Digitization} \label{sec:comp-with-digit}
Similarly to the notions of \emph{sampling} and \emph{sampling invariance} --- introduced in Section \ref{sec:sampling} --- the notions of \emph{digitization} and \emph{digitizability} \cite{HMP92} link dense- and discrete-time real-time semantics.
The main purpose of digitization is to provide a means to reduce the verification problem from the richer dense-time semantics to the simpler discrete-time one.
This section recalls the formal definition of digitization and digitizability and compares them against the notions of sampling, sampling invariance, and discrete-time approximations introduced in this paper.

There are two fundamental high-level differences between the frameworks of digitization and sampling; bridging them is necessary to carry out a formal comparison of the notions.
First, our framework considers dense- and discrete-time \emph{behaviors} as semantic structures, whereas digitization is defined for dense- and discrete-valued \emph{timed words}.
A timed word is a discrete sequence of timestamped events, such that every event is assumed to occur at the absolute time value of its timestamp.
Second, sampling invariance is a syntactic notion (i.e., it is a property that applies to \emph{formulas}), whereas digitizability is a semantic notion (i.e., it is a property that applies to \emph{sets of timed words}).
Let us introduce formally these ideas and the precise notions of digitization and digitizability.

\begin{definition}[MTL timed word semantics]
An (infinite) \emph{timed word} over $\alphabet$ is an $\omega$-sequence $(\sigma_0, t_0)(\sigma_1, t_1)\cdots(\sigma_i, t_i)\cdots$ in $(\alphabet \times \timedomain)^\omega$, such that the sequence of timestamps $t_i$ is weakly monotonic and diverging.
According to whether $\timedomain$ is a dense (typically $\reals_{\geq 0}$) or discrete (typically $\naturals$) set, the timed words are named \emph{dense-} or \emph{discrete-valued}.

MTL \emph{semantics} over timed words is defined as expected: given a timed word $\rho$, a position $i \in \naturals$, and an MTL formula $\phi$, we write $\rho, i \models \phi$ iff $\rho$ satisfies $\phi$ at position $i$.
The definition of the modalities is: $\rho, i \models \untilMTL{I}{\phi_1, \phi_2}$ iff there exists $j \geq i$ such that $t_j \in t_i \oplus I$, $\rho, j \models \phi_2$, and $\rho, k \models \phi_1$ for all $i \leq k < j$; and $\rho, i \models \relMTL{I}{\phi_1, \phi_2}$ iff for all $j \geq i$ such that $t_j \in t_i \oplus I$, it is $\rho, j \models \phi_2$ or $\rho, k \models \phi_1$ for some $i \leq k < j$.
Then, $\rho \models \phi$ iff $\rho, i \models \phi$ for all $i \in \naturals$.\footnote{The digitization paper \cite{HMP92} assumed an initial satisfiability semantics, but we adopt a global satisfiability semantics to allow a uniform comparison with sampling invariance (see Section \ref{sec:mtl-semantics}); it should be clear that this is without loss of generality.}
Given a formula $\phi$, $\worldw{\phi}{\timedomain}$ denotes the set $\{ \rho \mid \rho \models \phi \}$ of $\timedomain$-valued timed words that satisfy $\phi$.
\end{definition}

\begin{definition}[Digitization and digitizability]
Given a timed word $\rho = \{ (\sigma_i, t_i) \mid i \in \naturals \}$ and a fractional value $0 \leq \epsilon < 1$, the \emph{$\epsilon$-digitization} of $\rho$ is defined as the discrete-valued timed word $[\rho]_\epsilon = \{ (\sigma_i, [t_i]_\epsilon) \mid i \in \naturals \}$, where $[t]_\epsilon$ is $\lfloor t \rfloor$ if $t \leq \lfloor t \rfloor + \epsilon$, and $\lceil t \rceil$ otherwise.
The \emph{digitization} of a set of timed words $\Pi$ is the set $[\Pi]$ of discrete-valued timed words defined as $\{ [\rho]_\epsilon \mid \rho \in \Pi \text{ and } 0 \leq \epsilon < 1 \}$, i.e., the set of all possible digitizations of words in $\Pi$.

A set of timed words $\Pi$ is: (1) \emph{closed under digitization} (c.u.d.) iff $\rho \in \Pi$ implies $[\{\rho\}] \subseteq \Pi$; (2) \emph{closed under inverse digitization} (c.u.i.d.) iff $[\{\rho\}] \subseteq \Pi$ implies $\rho \in \Pi$; (3) \emph{digitizable} iff it is c.u.d.\ and c.u.i.d.
Correspondingly, an MTL formula $\phi$ is c.u.d., c.u.i.d., or digitizable, iff $\worldw{\phi}{\reals_{\geq 0}}$ is.
\end{definition}

For digitizable properties, discrete-time verification completely captures dense-time verification; more precisely, if a system specification is closed under digitization, and the requirements are closed under inverse digitization, the problem of determining if the specification meets the requirements is perfectly reducible to the discrete-time case.
However, it is difficult to characterize a significant syntactic subset of MTL formulas that are digitizable, and in fact only a few examples are given in \cite{HMP92}.
Moreover, digitization exploits weakly-monotonic timed word to ensure that no dense-time event is lost when digitizing a dense-valued timed word; this is why no notion similar to non-Berkeleyness is introduced.

The following example shows that digitizability and sampling invariance define \emph{incomparable classes} of MTL formulas, i.e., there exist sampling invariant non-digitizable formulas, as well as digitizable non sampling-invariant formulas.
This demonstrates that the two notions have different angles, and it suggests that techniques for discrete-time verification of dense-time MTL formulas based on these two orthogonal notions may each have its own complementary strengths and weaknesses.

\begin{example}
For $h \in \naturals_{> 0}$, let $\nondig_h$ be the $\flatMTL$ formula $\pp \Rightarrow \diamondMTL{< h}{\qq}$.
Theorem \ref{th:sinv-flat} proves that $\nondig_h$ is s.i.
Let us show that $\nondig_h$ is instead not c.u.d., hence neither digitizable.
Take any timed word $\sigma = \cdots (\pp, k)(\qq, k + h - 1 + \mu)(\qq, k+h+\mu)  \cdots$ with $k \in \naturals$, $0 < \mu < 1$, and such that $\pp$ does not occur anywhere else.
Any $\epsilon$-digitization of $\sigma$ for $\epsilon < \mu$ has the form $[\sigma]_\epsilon = \cdots (\pp, k)(\qq, k + h)(\qq, k + h + 1)  \cdots$.
Hence $\nondig_h$ is not c.u.d.\ because $\sigma \models \nondig_h$ but $[\sigma]_\epsilon \not\models \nondig_h$ for any such $\epsilon$.

For $h \in \naturals_{> 0}$, let $\nonsi_h$ be the MTL formula $\Som{\pp \wedge \nowonMTL{\neg \pp}} \wedge \psi_h$, where $\psi_h$ has been defined in Example \ref{ex:noncus}.
It is not difficult to show that $\Som{\pp \wedge \nowonMTL{\neg \pp}}$ is unsatisfiable in the timed word semantics, hence $\nonsi_h$ is trivially digitizable.
Let us show that $\nonsi_h$ is instead not c.u.s., hence neither s.i.
Take the same behavior $c \in \bsr_h$ of Example \ref{ex:noncus}, where we further assume that $V$ is a right-closed interval (see Figure \ref{fig:noncus}).
$c \mdr \nonsi_h$ because Example \ref{ex:noncus} showed that $c \mdr \psi_h$ and $\pp \wedge \nowonMTL{\neg \pp}$ holds at the right end-point of $V$.
However, Example \ref{ex:noncus} also proved that $\sigma_{h,z}[c] \not\mdz \eta_h^\reals \left[\psi_h\right]$, so $\sigma_{h,z}[c] \not\mdz \eta_h^\reals \left[\nonsi_h\right]$ as well.
Hence, $\nonsi_h$ is not c.u.s.
\end{example}

\subsection{Other Work on the Relations between Dense and Discrete Time} \label{sec:other-related}
The introduction of the notion of \emph{digitization} has spawned much derivative work, where the notion is applied to various formalisms.
Several authors considered digitization for automata-based real-time formalisms, especially timed automata \cite{BER94,Bos99,MP95,BMT99,BLN03,OW03,CLT07}.
Others studied how the decidability and complexity of standard verification problems for \emph{timed automata} (esp.~reachability) change when moving from a dense- to a discrete-time semantics, such as in \cite{GPV94,KP05}.
Asarin, Maler, and Pnueli \cite{AMP98} investigated instead to what extent \emph{qualitative} behavior of digital circuits (which can in turn be modeled as timed automata \cite{MP95}) is preserved in a sampled discrete-time semantics.
The focus of all these works is to determine to what extent the computationally simpler discrete-time semantics can be substituted for the dense-time semantics for automated verification.

The notion of digitization has been applied also to descriptive notations, such as real-time temporal logics and process algebras.
In the latter category, Ouaknine studies digitization for \emph{timed CSP} \cite{Oua02}; his main contribution is the proof that all CSP are closed under inverse digitization, hence they can be model-checked over dense time by considering just their discrete-time semantics.

Among temporal logics, the digitization of \emph{duration calculus} (DC) and its variants has been studied in several works.
Van Hung and Giang consider standard duration calculus and a slight generalization of digitization called \emph{sampling} \cite{HG96}.
Their work is focused on providing inference rules that allow one to infer the validity of dense-time formulas from the validity of sampled discrete-time formulas and \emph{vice versa}.
Another similarity with our approach is that they consider $\delta$-stability: a constraint similar to non-Berkeleyness that relates the ``speed'' of signals and the sampling period $\delta$.
Unlike non-Berkeleyness $\delta$-stability is asymmetric, in that whenever a proposition switches to true it must hold its truth value for more than $\delta$ time units, but it is not required to do so when it switches to false.

Pandya et al.~also have applied the notion of digitization to DC, with the aim of developing efficient dense-time verification techniques based on discretization.
Their overall approach consists of two parts, and it has been shown to be applicable to MTL as well \cite{Pandya-personal}.
In the first part \cite{CP03}, the notion of digitization has been applied to IDL (Interval Duration Logic) a DC variant whose formulas are interpreted over timed words.
Given that a syntactic characterization of closure under inverse digitization for IDL formulas is hard to achieve, a new notion of \emph{strong closure under inverse digitization} (SCID) is introduced.
SCID eases the problem because it is straightforward to determine if an IDL formula is SCID, and SCID entails closure under inverse digitization in the standard sense.
For formulas that are not SCID, approximations of formulas are introduced.
In the second part \cite{PNL07}, the richer semantics of DC (based on behaviors) is reduced to the timed word semantics of IDL through two approximation mappings $\alpha^+$ and $\alpha^-$.
$\alpha^+$ and $\alpha^-$ play a role similar to our over- and under- approximations $\overap{},\underap{}$, in that $\alpha^+$ preserves non-validity and $\alpha^-$ preserves validity from the sampled to the dense-time semantics.
Unsurprisingly, the resulting verification technique is incomplete, as DC is undecidable over dense time.

De Alfaro and Manna considered the problem of discretization for the predicate temporal logic TL \cite{dAM95}.
Their results are based on the semantic notion of \emph{finite variability}: informally, a formula $\phi$ is finitely variable if, for any timed word, one can find a refined ``ground'' timed word such that any subformula of $\phi$ has a constant truth value within any interval of the refined word.
For finitely variable formulas over ground traces, the satisfaction relation of a formula $\phi$ in the dense-time semantics corresponds to that of $\Omega(\phi)$ in the discrete-time semantics (where $\Omega$ is a given translation function).
Some sufficient syntactic conditions for a formula to achieve the finite variability requirement are introduced; based on these, a methodology for dense-time verification through refinement to discrete time is proposed.

Fainekos and Pappas \cite{FP07,FP07-unpub} present a technique for testing specifications written in MITL (an MTL subset) against continuous-time signals by analyzing only discrete samplings of the signals.
Their technique shares underlying motivations and ideas with ours, although the two approaches have complementary scopes: our results bridge the gap between the dense-time non-Berkeley semantics and the discrete-time semantics for MTL, whereas Fainekos and Pappas discover concrete and practical conditions under which the continuous-time behavior of a dynamical system can be analyzed by means of its discrete-time observations.

\subsection{The Sampling Theorem} \label{sec:sampling-theorem}
The sampling theorem \cite{sampling-book} states sufficient conditions for which no information loss occurs in the digital sampling of a continuous-time signal.
A continuous-time \emph{signal} $s$ is a mapping $s: \reals \rightarrow D$ where $D$ is some --- usually dense --- codomain.
$B_s$ denotes the \emph{bandwidth} of $s$, that is its highest frequency in $s$.\footnote{The highest frequency is defined as the largest nonzero value for which the Fourier transform $F[s]$ of $s$ is non-zero.}
Using the notation of Section \ref{sec:sampling}, the \emph{sampling} of $s$ with sampling period $\delta$ is the discrete-time signal $\sigma_{\delta, 0}[s]$.
The sampling theorem states that $s$ can be perfectly reconstructed from $\sigma_{\delta, 0}[s]$ for any $\delta < 1/(2 B_s)$.

A number of similarities between this fundamental theorem of signal theory and the results of this paper are apparent.
In particular, the requirement on the relation between bandwidth and sampling period is reminiscent of the non-Berkeleyness requirement, so that the results of this paper might seem a consequence of the sampling theorem.
Our dense-time behaviors $\bsr$ can indeed be modeled as continuous-time signals over range $[0,2^{|\alphabet|}]$.
However all of them have \emph{infinite bandwidth} because of the discontinuities corresponding to transition points, regardless of whether they are non-Berkeley or not.
Hence the sampling theorem cannot strictly be applied to Boolean-valued signals.
Nonetheless, a connection between the theory of sampling and the theory of this paper exists, as we demonstrate in the following.

\begin{example}
Consider a simple unary alphabet $\{\pp\}$ and a single behavior $b$ such that $\pp$ holds over $\reals_{> 0}$ and does not hold over $\reals_{< 0}$ (we disregard the value of $\pp$ exactly at $0$).
$b$ corresponds to the signal $s:\reals \rightarrow [0,1]$ defined as $s(t) = H(t)$ where $H$ denotes the usual (Heaviside) unit step function (see Figure \ref{fig:samp-th}).
\begin{figure}[!h]
  \centering
  \includegraphics[scale=.6]{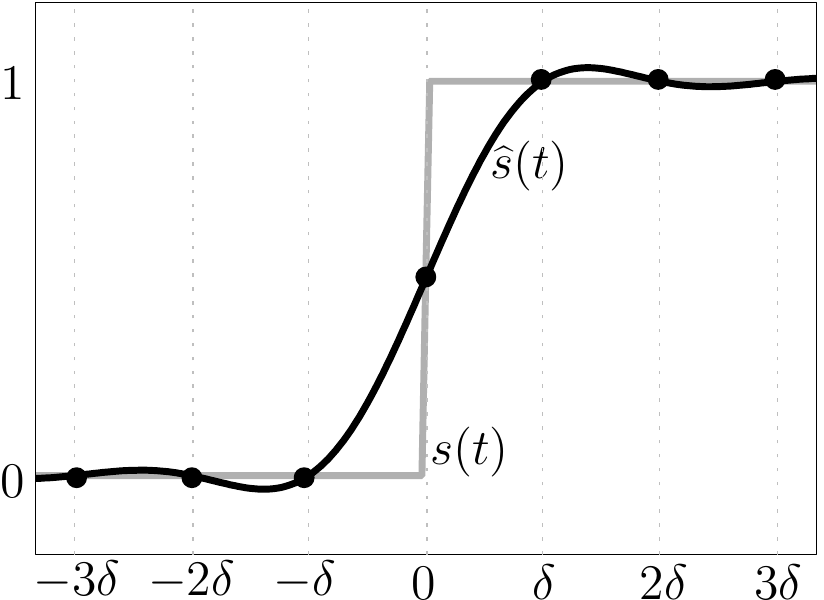}
  \caption{Signals $s(t)$ (in gray) and $\widehat{s}(t)$ (in black).}
  \label{fig:samp-th}
\end{figure}
$b$ can also be described perfectly by the MTL formula $\beta = \boxPMTL{> 0}{\neg \pp} \wedge \boxMTL{> 0}{\pp}$ evaluated at the origin.
The discrete-time MTL formula $\beta' = \adapt{\beta} = \boxPMTL{\geq 1}{\neg \pp} \wedge \boxMTL{\geq 1}{\pp}$ characterizes discrete-time samplings of $b$ according to our theory.
$\beta'$ can be seen as describing some dense-time behaviors in $\bsrd$ through their samplings: all behaviors such that $\pp$ holds over $\reals_{\geq \delta}$ and it does not hold over $\reals_{\leq -\delta}$.
Hence, the sampling has introduced an information loss in the formula about where exactly $\pp$ switches within $(-\delta, \delta)$.
If we try to reconstruct $s$ from its digital sampling according to the classical theory, we notice that we introduce a similar information loss.
In fact, let $\widehat{s}: \reals \rightarrow [0,1]$ be the continuous-time reconstruction of $\sigma_{\delta, 0}[s]$ built with the Whittaker-Shannon interpolation formula, i.e., $\widehat{s}(t) = \sum_{k \in \integers} \sigma_{\delta, 0}[s](k) \sinc((t - k\delta)/\delta)$.
As it can be seen in Figure \ref{fig:samp-th}, $\widehat{s}$ coincides almost perfectly with $s$ over $\reals_{\leq -\delta} \cup \reals_{\geq \delta}$ (the residual errors are only due to numerical approximations), whereas it deviates significantly within $(-\delta, \delta)$ due to the information loss introduced with sampling (it passes right through the origin only as a result of symmetry).
In this sense information loss for Boolean-valued signals are similar in our theory for MTL and in classical sampling theory for signals.
\end{example}


%% file: conclusion.tex
\section{Conclusion} \label{sec:conclusion}
In this paper, we presented an approach to relate dense-time MTL formulas to some discrete-time counterparts (and \emph{vice versa}).
We exploited the resulting relationship to define a technique for the verification through discretization of systems described as dense-time MTL formulas.
The verification technique is inherently incomplete, though in practice it has yielded promising results \cite{FPR08-FM08,FPR08-ICFEM08,BFPR09-SEFM09}.

In the future, we plan to apply the notion of sampling presented in this paper to the synthesis of software components of real-time systems from continuous-time specifications.
We will also further investigate the properties of the verification technique presented in Section \ref{sec:verification}, in particular to better characterize, and possibly reduce, the scope of its incompleteness.

%% file: samplingMTL.bbl
\newcommand{\etalchar}[1]{$^{#1}$}

%% file: samplingMTL.bbl
\begin{thebibliography}{BMOW07}

\bibitem[AFH96]{AFH96}
Rajeev Alur, Tom{\' a}s Feder, and Thomas~A. Henzinger.
\newblock The benefits of relaxing punctuality.
\newblock {\em Journal of the ACM}, 43(1):116--146, 1996.

\bibitem[AH93]{AH93}
Rajeev Alur and Thomas~A. Henzinger.
\newblock Real-time logics: Complexity and expressiveness.
\newblock {\em Information and Computation}, 104(1):35--77, 1993.

\bibitem[AMP98]{AMP98}
Eugene Asarin, Oded Maler, and Amir Pnueli.
\newblock On discretization of delays in timed automata and digital circuits.
\newblock In Davide Sangiorgi and Robert de~Simone, editors, {\em Proceedings
  of the 9th International Conference on Concurrency Theory (CONCUR'98)},
  volume 1466 of {\em Lecture Notes in Computer Science}, pages 470--484.
  Springer-Verlag, 1998.

\bibitem[BER94]{BER94}
Ahmed Bouajjani, Rachid Echahed, and Riadh Robbana.
\newblock Verifying invariance properties of timed systems with duration
  variables.
\newblock In {\em Proceedings of the 3rd International Symposium on Formal
  Techniques in Real-Time and Fault-Tolerant Systems (FTRTFT'94)}, volume 863
  of {\em Lecture Notes in Computer Science}, pages 193--210. Springer-Verlag,
  1994.

\bibitem[BF01]{sampling-book}
John~J. Benedetto and Paulo J. S.~G. Ferreira, editors.
\newblock {\em Modern Sampling Theory}.
\newblock Birk{\" a}user Boston, 2001.

\bibitem[BFPR09]{BFPR09-SEFM09}
Marcello~M. Bersani, Carlo~A. Furia, Matteo Pradella, and Matteo Rossi.
\newblock Integrated modeling and verification of real-time systems through
  multiple paradigms.
\newblock In {\em Proceedings of the 7th IEEE International Conference on
  Software Engineering and Formal Methods (SEFM'09)}. IEEE Computer Society
  Press, November 2009.

\bibitem[BLN03]{BLN03}
Dirk Beyer, Claus Lewerentz, and Andreas Noack.
\newblock Rabbit: A tool for {BDD}-based verification of real-time systems.
\newblock In Warren A.~Hunt Jr. and Fabio Somenzi, editors, {\em Proceedings of
  the 15th International Conference on Computer Aided Verification (CAV'03)},
  volume 2725 of {\em Lecture Notes in Computer Science}, pages 122--125.
  Springer-Verlag, 2003.

\bibitem[BMOW07]{BMOW07}
Patricia Bouyer, Nicolas Markey, Jo{\"e}l Ouaknine, and James Worrell.
\newblock The cost of punctuality.
\newblock In {\em Proceedings of the 22nd IEEE Symposium on Logic in Computer
  Science (LICS'07)}. IEEE Computer Society, 2007.

\bibitem[BMP{\etalchar{+}}07]{BMPSS07}
Domenico Bianculli, Angelo Morzenti, Matteo Pradella, Pierluigi {San Pietro},
  and Paola Spoletini.
\newblock {Trio2Promela}: A model checker for temporal metric specifications.
\newblock In {\em ICSE Companion}, pages 61--62, 2007.

\bibitem[BMT99]{BMT99}
Marius Bozga, Oded Maler, and Stavros Tripakis.
\newblock Efficient verification of timed automata using dense and discrete
  time semantics.
\newblock In Laurence Pierre and Thomas Kropf, editors, {\em Proceedings of the
  10th Correct Hardware Design and Verification Methods Advanced Research
  Working Conference (CHARME'99)}, volume 1703 of {\em Lecture Notes in
  Computer Science}, pages 125--141. Springer-Verlag, 1999.

\bibitem[Bo{\v s}99]{Bos99}
Dragan Bo{\v s}na{\v c}ki.
\newblock Digitization of timed automata.
\newblock In {\em Proceedings of the 4th International Workshop on Formal
  Methods for Industrial Critical Systems (FMICS'99)}, pages 283--302, 1999.

\bibitem[CC00]{CC00}
Hubert Comon and V{\'e}ronique Cortier.
\newblock Flatness is not a weakness.
\newblock In {\em Proceedings of the 14th Annual Conference of the {EACSL} on
  Computer Science Logic}, volume 1862 of {\em Lecture Notes in Computer
  Science}, pages 262--276. Springer-Verlag, 2000.

\bibitem[CCG{\etalchar{+}}02]{NuSMV2}
Alessandro Cimatti, Edmund~M. Clarke, Enrico Giunchiglia, Fausto Giunchiglia,
  Marco Pistore, Marco Roveri, Roberto Sebastiani, and Armando Tacchella.
\newblock {NuSMV} 2: An opensource tool for symbolic model checking.
\newblock In {\em Proceeding of the 14th International Conference on
  Computer-Aided Verification (CAV'02)}, volume 2404 of {\em Lecture Notes in
  Computer Science}, pages 359--364. Springer-Verlag, 2002.

\bibitem[CLT07]{CLT07}
Edmund~M. Clarke, Flavio Lerda, and Muralidhar Talupur.
\newblock An abstraction technique for real-time verification.
\newblock In {\em Proceedings of the GM R\&D Workshop on Next Generation Design
  and Verification Methodologies for Distributed Embedded Control System},
  2007.

\bibitem[CP03]{CP03}
Gaurav Chakravorty and Paritosh~K. Pandya.
\newblock Digiziting interval duration logic.
\newblock In Warren~A. {Hunt, Jr.} and Fabio Somenzi, editors, {\em Proceedings
  of the 15th International Conference on Computer Aided Verification
  (CAV'03)}, volume 2725 of {\em Lecture Notes in Computer Science}, pages
  167--179. Springer-Verlag, 2003.

\bibitem[Dam99]{Dam99}
Dennis Dams.
\newblock Flat fragments of {CTL} and {CTL*}: Separating the expressive and
  distinguishing powers.
\newblock {\em Logic Journal of the IGPL}, 7(1):55--78, 1999.

\bibitem[DDMR09]{Alaska}
Martin {De Wulf}, Laurent Doyen, Nicolas Maquet, and Jean-Fran{\c c}ois Raskin.
\newblock {ALASKA}: {A}ntichains for {L}ogic, {A}utomata and {S}ymbolic
  {K}ripke structures {A}nalysis.
\newblock In {\em Proceeding of the 6th International Symposium on Automated
  Technology for Verification and Analysis (ATVA'08)}, volume 5311 of {\em
  Lecture Notes in Computer Science}, pages 240--245. Springer-Verlag, 2009.

\bibitem[dM95]{dAM95}
Luca {de Alfaro} and Zohar Manna.
\newblock Verification in continuous time by discrete reasoning.
\newblock In Vangalur~S. Alagar and Maurice Nivat, editors, {\em Proceedings of
  the 4th International Conference on Algebraic Methodology and Software
  Technology (AMAST'95)}, volume 936 of {\em Lecture Notes in Computer
  Science}, pages 292--306. Springer-Verlag, 1995.

\bibitem[DMP07]{MP07}
Deepak D'Souza, Raj {Mohan M.}, and Pavithra Prabhakar.
\newblock Flattening metric temporal logic.
\newblock Manuscript, 2007.

\bibitem[DS02]{DS02}
St{\'e}phane Demri and Philippe Schnoebelen.
\newblock The complexity of propositional linear temporal logics in simple
  cases.
\newblock {\em Information and Computation}, 174(1):84--103, 2002.

\bibitem[EW96]{EW96}
Kousha Etessami and Thomas Wilke.
\newblock An until hierarchy for temporal logic.
\newblock In {\em Proceedings of the 11th Annual IEEE Symposium on Logic in
  Computer Science (LICS'96)}, pages 108--117. IEEE Computer Society Press,
  1996.

\bibitem[FMMR10]{FMMR08}
Carlo~A. Furia, Dino Mandrioli, Angelo Morzenti, and Matteo Rossi.
\newblock Modeling time in computing: a taxonomy and a comparative survey.
\newblock {\em ACM Computing Surveys}, 42(2):1--59, February 2010.
\newblock Article 6. Also available as http://arxiv.org/abs/0807.4132.

\bibitem[FP07a]{FP07}
Georgios~E. Fainekos and George~J. Pappas.
\newblock Robust sampling for {MITL} specifications.
\newblock In {\em Proceedings of the 5th International Conference on Formal
  Modelling and Analysis of Timed Systems (FORMATS'07)}, volume 4763 of {\em
  Lecture Notes in Computer Science}, pages 147--162. Springer-Verlag, October
  2007.

\bibitem[FP07b]{FP07-unpub}
Georgios~E. Fainekos and George~J. Pappas.
\newblock Robustness of temporal logic specifications for continuous time
  signals.
\newblock Submitted, November 2007.

\bibitem[FPR08a]{FPR08-FM08}
Carlo~A. Furia, Matteo Pradella, and Matteo Rossi.
\newblock Automated verification of dense-time {MTL} specifications via
  discrete-time approximation.
\newblock In Jorge Cu{\'e}llar and Tom Maibaum, editors, {\em Proceedings of
  the 15th International Symposium on Formal Methods (FM'08)}, volume 5014 of
  {\em Lecture Notes in Computer Science}, pages 132--147. Springer-Verlag, May
  2008.

\bibitem[FPR08b]{FPR08-ICFEM08}
Carlo~A. Furia, Matteo Pradella, and Matteo Rossi.
\newblock Practical automated partial verification of multi-paradigm real-time
  models.
\newblock In Shaoying Liu, Tom Maibaum, and Keijiro Araki, editors, {\em
  Proceedings of the 10th International Conference on Formal Engineering
  Methods (ICFEM'08)}, volume 5256 of {\em Lecture Notes in Computer Science},
  pages 298--317. Springer-Verlag, October 2008.

\bibitem[FR06]{FR06}
Carlo~A. Furia and Matteo Rossi.
\newblock Integrating discrete- and continuous-time metric temporal logics
  through sampling.
\newblock In Eugene Asarin and Patricia Bouyer, editors, {\em Proceedings of
  the 4th International Conference on Formal Modelling and Analysis of Timed
  Systems (FORMATS'06)}, volume 4202 of {\em Lecture Notes in Computer
  Science}, pages 215--229. Springer-Verlag, September 2006.

\bibitem[FR07]{FR07-FORMATS07}
Carlo~A. Furia and Matteo Rossi.
\newblock On the expressiveness of {MTL} variants over dense time.
\newblock In Jean-Fran{\c c}ois Raskin and P.~S. Thiagarajan, editors, {\em
  Proceedings of the 5th International Conference on Formal Modelling and
  Analysis of Timed Systems (FORMATS'07)}, volume 4763 of {\em Lecture Notes in
  Computer Science}, pages 163--178. Springer-Verlag, October 2007.

\bibitem[FR08]{FR08-FORMATS08}
Carlo~A. Furia and Matteo Rossi.
\newblock {MTL} with bounded variability: Decidability and complexity.
\newblock In Franck Cassez and Claude Jard, editors, {\em Proceedings of the
  6th International Conference on Formal Modelling and Analysis of Timed
  Systems (FORMATS'08)}, volume 5215 of {\em Lecture Notes in Computer
  Science}, pages 109--123. Springer-Verlag, September 2008.

\bibitem[Fur07]{Fur07}
Carlo~Alberto Furia.
\newblock {\em Scaling up the formal analysis of real-time systems}.
\newblock PhD thesis, Dipartimento di Elettronica e Informazione, Politecnico
  di Milano, May 2007.

\bibitem[GKP94]{GKP94}
Ronald~L. Graham, Donald~E. Knuth, and Oren Patashnik.
\newblock {\em Concrete Mathematics: A foundation for computer science}.
\newblock Addison-Wesley, 2nd edition, 1994.

\bibitem[GPV94]{GPV94}
Aleks G{\" o}ll{\" u}, Anuj Puri, and Pravin Varaiya.
\newblock Discretization of timed automata.
\newblock In {\em Proceedings of the 33rd Conference on Decision and Control},
  pages 957--958, 1994.

\bibitem[HG96]{HG96}
Dang~Van Hung and Phan~Hong Giang.
\newblock Sampling semantics of duration calculus.
\newblock In Bengt Jonsson and Joachim Parrow, editors, {\em Proceedings of the
  4th International Symposium on Formal Techniques in Real-Time and
  Fault-Tolerant Systems (FTRTFT'96)}, volume 1135 of {\em Lecture Notes in
  Computer Science}, pages 188--207. Springer-Verlag, 1996.

\bibitem[HM96]{HM96}
Constance Heitmeier and Dino Mandrioli, editors.
\newblock {\em Formal Methods for Real-Time Computing}.
\newblock John Wiley \& Sons, 1996.

\bibitem[HMP92]{HMP92}
Thomas~A. Henzinger, Zohar Manna, and Amir Pnueli.
\newblock What good are digital clocks?
\newblock In Werner Kuich, editor, {\em Proceedings of the 19th International
  Colloquium on Automata, Languages and Programming (ICALP'92)}, volume 623 of
  {\em Lecture Notes in Computer Science}, pages 545--558. Springer-Verlag,
  1992.

\bibitem[HR04]{HR04}
Yoram Hirshfeld and Alexander~Moshe Rabinovich.
\newblock Logics for real time: Decidability and complexity.
\newblock {\em Fundamenta Informaticae}, 62(1):1--28, 2004.

\bibitem[HR05]{HR05}
Yoram Hirshfeld and Alexander~Moshe Rabinovich.
\newblock Timer formulas and decidable metric temporal logic.
\newblock {\em Information and Computation}, 198(2):148--178, 2005.

\bibitem[HS06]{HS06}
Thomas~A. Henzinger and Joseph Sifakis.
\newblock The embedded systems design challenge.
\newblock In Jayadev Misra, Tobias Nipkow, and Emil Sekerinski, editors, {\em
  Proceedings of the 14th International Symposium on Formal Methods (FM'06)},
  volume 4085 of {\em Lecture Notes in Computer Science}, pages 1--15.
  Springer-Verlag, 2006.

\bibitem[Koy90]{Koy90}
Ron Koymans.
\newblock Specifying real-time properties with metric temporal logic.
\newblock {\em Real-Time Systems}, 2(4):255--299, 1990.

\bibitem[Koy92]{Koy92}
Ron Koymans.
\newblock (real) time: A philosophical perspective.
\newblock In J.~W. de~Bakker, Cornelis Huizing, Willem~P. de~Roever, and
  Grzegorz Rozenberg, editors, {\em Proceedings of the REX Workshop:
  ``Real-Time: Theory in Practice''}, volume 600 of {\em Lecture Notes in
  Computer Science}, pages 353--370. Springer-Verlag, 1992.

\bibitem[KP05]{KP05}
Pavel Kr{\v c}{\' a}l and Radek Pel{\' a}nek.
\newblock On sampled semantics of timed systems.
\newblock In R.~Ramanujam and Sandeep Sen, editors, {\em Proceedings of the
  25th International Conference on Foundations of Software Technology and
  Theoretical Computer Science (FSTTCS'05)}, volume 3821 of {\em Lecture Notes
  in Computer Science}, pages 310--321. Springer-Verlag, 2005.

\bibitem[KS05]{KS05}
Anton{\' i}n Ku{\v c}era and Jan Strej{\v c}ek.
\newblock The stuttering principle revisited.
\newblock {\em Acta Informatica}, 41(7/8):415--434, 2005.

\bibitem[MNP06]{MNP06}
Oded Maler, Dejan Nickovic, and Amir Pnueli.
\newblock From {MITL} to timed automata.
\newblock In Eugene Asarin and Patricia Bouyer, editors, {\em Proceedings of
  the 4th International Conference on Formal Modeling and Analysis of Timed
  Systems (FORMATS'06)}, volume 4202 of {\em Lecture Notes in Computer
  Science}, pages 274--289. Springer-Verlag, 2006.

\bibitem[MP95]{MP95}
Oded Maler and Amir Pnueli.
\newblock Timing analysis of asynchronous circuits using timed automata.
\newblock In Paolo Camurati and Hans Eveking, editors, {\em Proceedings of the
  Advanced Research Working Conference on Correct Hardware Design and
  Verification Methods}, volume 987 of {\em Lecture Notes in Computer Science},
  pages 189--205. Springer-Verlag, 1995.

\bibitem[Oua02]{Oua02}
Jo{\" e}l Ouaknine.
\newblock Digitisation and full abstraction for dense-time model checking.
\newblock In Joost-Pieter Katoen and Perdita Stevens, editors, {\em Proceedings
  of the 8th International Conference on Tools and Algorithms for the
  Construction and Analysis of Systems (TACAS'02)}, volume 2280 of {\em Lecture
  Notes in Computer Science}, pages 37--51. Springer-Verlag, 2002.

\bibitem[OW03]{OW03}
Jo{\" e}l Ouaknine and James Worrell.
\newblock Revisiting digitization, robustness, and decidability for timed
  automata.
\newblock In {\em Proceedings of the 18th Annual IEEE Symposium on Logic in
  Computer Science (LICS'03)}, pages 198--207. IEEE Computer Society Press,
  2003.

\bibitem[Pan08]{Pandya-personal}
Paritosh~K. Pandya.
\newblock Personal communication, September 2008.

\bibitem[PMS07]{PMS07}
Matteo Pradella, Angelo Morzenti, and Pierluigi {San Pietro}.
\newblock The symmetry of the past and of the future: Bi-infinite time in the
  verification of temporal properties.
\newblock In {\em Proceedings of The 6th joint meeting of the European Software
  Engineering Conference and the ACM SIGSOFT Symposium on the Foundations of
  Software Engineering (ESEC/FSE 2007)}, pages 312--320, 2007.

\bibitem[PNL07]{PNL07}
Paritosh~K. Pandya, Shankara {Narayanan Krishna}, and Kuntal Loya.
\newblock On sampling abstraction of continuous time logic with durations.
\newblock In {\em Proceeding of the 13th International Conference on Tools and
  Algorithms for the Construction and Analysis of Systems (TACAS'07)}, volume
  4424 of {\em Lecture Notes in Computer Science}, pages 246--260.
  Springer-Verlag, 2007.

\bibitem[PP04]{PP04}
Dominique Perrin and Jean-{\'E}ric Pin.
\newblock {\em Infinite Words}, volume 141 of {\em Pure and Applied
  Mathematics}.
\newblock Elsevier, 2004.

\bibitem[PSSM03]{PSSM03}
Matteo Pradella, Pierluigi {San Pietro}, Paola Spoletini, and Angelo Morzenti.
\newblock Practical model checking of {LTL} with past.
\newblock In Farn Wang and Insup Lee, editors, {\em Proceedings of 1st
  International Workshop on Automated Technology for Verification and Analysis
  (ATVA'03)}, pages 135--146, Taipei, Taiwan, R.O.C., December 2003.

\bibitem[Rab03]{Rab03}
Alexander~Moshe Rabinovich.
\newblock Automata over continuous time.
\newblock {\em Theoretical Computer Science}, 300(1--3):331--363, 2003.

\bibitem[TW04]{TW04}
Denis Th{\' e}rien and Thomas Wilke.
\newblock Nesting until and since in linear temporal logic.
\newblock {\em Theory of Computing Systems}, 37(1):111--131, 2004.

\end{thebibliography}
